\documentclass[a4paper,UKenglish,cleveref, autoref, thm-restate]{lipics-v2021}

\hideLIPIcs


\usepackage{hyperref}
\hypersetup{
  citecolor=blue,
  colorlinks=true,
  pdftitle={Böhm-Taylor-dBang}
}

\usepackage{fbox}

\usepackage{amsmath,amssymb}
\usepackage{amsfonts}
\usepackage[utf8]{inputenc}
\usepackage[T1]{fontenc}
\usepackage{thmtools}
\usepackage{amsthm}

\usepackage{tikz}
\usetikzlibrary{calc}
\usetikzlibrary{arrows.meta}
\usetikzlibrary{shapes, patterns, decorations.pathreplacing, decorations.text}
\usetikzlibrary{trees}
\usetikzlibrary{positioning}
\usetikzlibrary{tikzmark, cd}

\usepackage{longtable}
\usepackage{proof}
\usepackage{bussproofs}
\usepackage{mathpartir}
\usepackage{dsfont}
\usepackage{amsfonts}
\usepackage{stmaryrd}
\usepackage{mathtools}
\usepackage{cmll}
\usepackage{scalerel}

\usepackage{pdflscape}

\usepackage{graphicx}
\usepackage{color}
\usepackage{enumitem}

\usepackage{longtable}
\usepackage{multirow}

\newif\iflong
\longtrue

\title{Approximation theory for distant Bang calculus}
\ArticleNo{1}
\author{Kostia Chardonnet}
{Université de Lorraine, CNRS, Inria, LORIA, F-54000 Nancy, France \url{https://kostiachardonnet.github.io/}}
{kostia.chardonnet@pm.me}
{0009-0000-0671-6390} 
{}

\author{Jules Chouquet}
{Université d'Orléans, INSA CVL, LIFO, UR 4022, Orléans, France \url{https://www.univ-orleans.fr/lifo/membres/chouquet/}}
{jules.chouquet@univ-orleans.fr}
{0000-0003-2676-0297}
{}

\author{Axel Kerinec}
{Université Paris Est Creteil, LACL, F-94010 Créteil, France \url{https://axelkrnc.github.io/}}
{axel.kerinec@yahoo.com}
{0000-0003-0920-8847}
{}

\funding{This work is supported by the Plan
France 2030 through the PEPR integrated project EPiQ ANR-22-PETQ-0007 and the
HQI platform ANR-22-PNCQ-0002; and by the European project MSCA Staff Exchanges
Qcomical HORIZON-MSCA-2023-SE-01. The project is also supported by the Maison du
Quantique MaQuEst.}

\iflong \else
	\relatedversion{}\relatedversiondetails[cite={chardonnet2026approximationtheorydistantbang}]{Full Version}{https://arxiv.org/abs/2601.05199}
\fi

\authorrunning{K. Chardonnet and J. Chouquet and A. Kerinec} 

\Copyright{Kostia Chardonnet and Jules Chouquet and Axel Kerinec} 


\usepackage{todonotes}

\presetkeys%
    {todonotes}%
    {inline,backgroundcolor=yellow, bordercolor=red}{}

\newcommand{\<}{\langle}
\renewcommand{\>}{\rangle}

\newcommand{\alt}{~\mid~}
\newcommand{\set}[1]{\ensuremath{\{#1\}}}
\newcommand{\der}[1]{\ensuremath{\mathbf{der}(#1)}}
\renewcommand{\deg}[2]{d_{#1}(#2)}

\newcommand{\proc}{\ensuremath{\mathtt{Proc}}}

\newcommand{\fv}{\mathbf{fv}}

\newcommand{\dcbn}{\ensuremath{\mathtt{dCbN}}}
\newcommand{\cbn}{\ensuremath{\mathtt{CbN}}}

\newcommand{\dcbv}{\ensuremath{\mathtt{dCbV}}}
\newcommand{\cbv}{\ensuremath{\mathtt{CbV}}}

\newcommand{\bang}{\ensuremath{\mathtt{Bang}}}
\newcommand{\dbang}{\ensuremath{\mathtt{dBang}}}
\newcommand{\dbangb}{\ensuremath{\dbang_\bot}}

\newcommand{\dbangres}{\ensuremath{\mathtt{\delta Bang}}}

\newcommand{\dbangv}{\ensuremath{\dbang_V}}
\newcommand{\dbangn}{\ensuremath{\dbang_N}}

\newcommand{\nf}[1]{\mathbf{nf}(#1)}

\newcommand{\ton}{\ensuremath{\to_{\text{n}}}}
\newcommand{\tov}{\ensuremath{\to_{\text{v}}}}

\newcommand{\todb}{\mapsto_{\oc}} 
\newcommand{\todbs}{\to_{\oc s}} 
\newcommand{\todbf}{\to_{\oc}} 

\newcommand{\Todbr}{\Rightarrow_{\delta}} 
\newcommand{\todbr}{\to_{\delta}} 
\newcommand{\todbrs}{\to_{\delta s}} 
\newcommand{\todbrf}{\to_{\delta}} 
\newcommand{\totodbr}{\rightrightarrows_{\delta}} 

\newcommand{\extend}[1]{\overline{#1}}

\newcommand{\es}[2]{[#1/#2]}
\newcommand{\s}[2]{\{#1/#2\}}

\newcommand{\tradv}[1]{\ensuremath{{#1}^v}}
\newcommand{\tradn}[1]{\ensuremath{{#1}^n}}

\newcommand{\app}{\triangleleft_\oc} 
\newcommand{\appv}{\triangleleft_v} 
\newcommand{\appn}{\triangleleft_n} 

\newcommand{\mcal}[1]{\ensuremath{\mathcal{#1}}}
\newcommand{\setApprox}{\mcal{A}}
\newcommand{\setapprox}{\setApprox}

\newcommand{\taylor}[1]{\ensuremath{\mathcal{T}(#1)}}

\newcommand{\taylorn}[1]{\ensuremath{\mathcal{T}^n(#1)}}
\newcommand{\taylorv}[1]{\ensuremath{\mathcal{T}^v(#1)}}
\newcommand{\taylornf}[1]{\ensuremath{\text{TNF}(#1)}}
\newcommand{\bt}{\ensuremath{{\tt{AT}}}}
\newcommand{\btb}{\ensuremath{{\tt{AT}_\oc}}}
\newcommand{\btn}{\ensuremath{{\tt{AT}_n}}}
\newcommand{\btv}{\ensuremath{{\tt{AT}_v}}}
\newcommand{\smaller}{\ensuremath{\sqsubseteq}}

\newcommand{\fnl}[1]{\textcolor{black}{#1}}

\usepackage{thmtools}

\usepackage{lineno}
\iflong \nolinenumbers \else \linenumbers \fi

\begin{document}

\maketitle

\begin{abstract}
  Approximation semantics capture the observable behaviour of $\lambda$-terms.
 Böhm Trees and Taylor Expansion are its two central paradigms, related by the
 Commutation Theorem. While these notions are well understood in Call-by-Name (\cbn{}), they have only recently been developed for Call-by-Value (\cbv{}), which motivate the search for
 a unified approximation framework. The \bang{}-calculus provides such a
 framework: it subsumes both \cbn{} and \cbv{} through linear-logic translations and enjoys robust rewriting properties. We develop the approximation semantics
 of \dbang{} (the \bang{}-calculus with explicit substitutions and distant
 reductions) by introducing approximation trees in the Böhm tradition together with Taylor expansion. We establish their
 fundamental properties, including a commutation theorem. Via translations, our results recover the \cbn{} and \cbv{}
 cases within a single unifying framework capturing infinitary and
 resource-sensitive semantics.
\end{abstract}

\section{Introduction}
A central approach to $\lambda$-calculus semantics is program approximation
theory, which captures program behaviour, finitarily or infinitarily, offering a
characterization of ``meaningful''
terms.  

In \emph{Call-by-Name (\cbn{})}, meaningful terms are the \emph{solvable}
ones, i.e. those reducing to identity under some``testing context''. In \emph{Call-by-Value (\cbv{})}, meaningfulness is given by \emph{scrutability}\footnote{Also called potential valuablility}, which only requires reduction to a value and is strictly finer. 
Among approximation techniques, \emph{evaluation trees} and \emph{Taylor Expansions} are the most influential. 
\emph{Böhm trees} introduced by Barendregt ~\cite{barendregt-bohm} are the seminal evaluation trees. They associate to each $\lambda$-term a (possibly infinite) tree whose nodes describe
successive approximations of the term's \emph{head normal form}, or $\bot$ —
canonical term meaning nonsense, or divergence — when none exists, thus capturing its asymptotic \cbn{} behaviour.
 Böhm trees were related to the notion of \emph{solvability} of  \cbn{}  terms by the fact that a term is solvable if and only if its Böhm tree is not $\bot$.

Ehrhard and Regnier later introduced Taylor expansion ~\cite{ehrhard-uniformity-taylor}, inspired by the differential
$\lambda$-calculus~\cite{ehrhard-differential} and relational
semantics~\cite{laird-weighted-rel}.
Taylor expansion unfolds a $\lambda$-term into an (infinite) formal sum of resource terms underlying the linear-use of resources during computations, rather than progressively revealing their shape as Böhm trees do.

The Commutation Theorem~\cite[Corollary
35]{ehrhard-uniformity-taylor} states that normalizing the Taylor expansion of a term yields exactly the Taylor expansion of its Böhm tree. Taylor expansions can then be understood as
a \emph{resource-sensitive} version of Böhm trees.
Originally proved for
the \cbn{} $\lambda$-calculus, this result provides a deep bridge between
infinitary semantics and differential/resource semantics.

While Taylor expansions for the \cbv{} calculus were already studied
in~\cite{alberto-Giulio-CbV-Solvability,ehrhard:Collapsing}, the development of
\cbv{} Böhm trees remained open until Kerinec's PhD
work~\cite{kerinec:bohmTreeCbV,kerinec:phd}. Those Böhm trees have the same
commutation theorem with Taylor expansion as in the \cbn{} case. They also respect
the same relation with scrutability as \cbn{} Böhm trees do with solvability. This
emphasizes that scrutability is the appropriate notion of meaningfulness in \cbv{}.
However, this result is obtained in an alternative version of the original
Plotkin \cbv{}, which is known to have issues due to the $\beta_v$-reduction being
"too weak". Concretely, unlike in \cbn{}, \cbv{} reduction may get stuck: redexes can
be blocked since their argument is in normal form but not a value. This
phenomenon prevents a straightforward infinitary unfolding analogous to the \cbn{}
case.\\

The aforementioned \cbv{} Böhm trees are defined in the
 \emph{$\lambda_v^\sigma$-calculus} from \cite{Carraro14}, where the
 $\beta_v$-reduction is extended with permutation rules, the so-called
 $\sigma$-rules, originating from the translation of $\lambda$-terms to
 proof-nets~\cite{Girard_Lafont_Regnier_1995}. 
Another way to solve the \cbv{} issue, also coming from proof-nets, is using a
distance-based \cbv{} calculus: \emph{distant \cbv{} (\dcbv{})}\footnote{Also called Value
Substitution Calculus} ~\cite{structural-lambda-calculus,
accattoli:LIPIcs.RTA.2012.6, Accattoli_2012}.  In this system, substitutions may
be frozen thanks to \emph{explicit
substitutions}, written $M\es N x$, which do \emph{not} correspond to
(effective) substitutions, but instead represent substitutions that are yet to be
evaluated. The rewriting rules then act \emph{at a distance} with respect to the
explicit substitutions. \\

The divergence between \cbn{} and \cbv{} has historically required separate developments of most semantic notions. 
Such a duplication naturally called for a unifying perspective. \emph{Call-by-Push-Value (CbPV)}~\cite{levy-cbpv} provides precisely this, reconciling typed \cbn{} and \cbv{} within a single calculus structured around a clear distinction between values and computations.
 CbPV later gave rise to its untyped analogue \emph{Bang-calculus (\bang{})}~\cite{ehrhard2016bang}, by use of Linear Logic~\cite{ehrhard-cbpv-ll}. Both \cbn{}
and \cbv{} arise within the \bang{}-calculus via Girard's translations of the
intuitionistic arrow into Linear Logic~\cite{girard-ll}, making it a natural
setting in which to seek a uniform approximation theory. The \bang{}-calculus is
then an extension of the $\lambda$-calculus with two new constructs: $\oc M$
(pronounced ``bang $M$'') which \emph{freezes} the computation of $M$, and
$\der{M}$ which unfreezes it. Both \cbn{} and \cbv{} can then be translated into the
\bang{}-calculus, which simulates their rewriting strategies within a single
rewriting system.

However, the \bang{}-calculus exhibits the same issue as \cbv{}: ill-formed redexes may
block evaluation. One might expect that adapting the $\sigma$-rules to the
\bang{}-calculus would solve this issue; however, while the \cbv{} calculus with
$\sigma$-rules is confluent, this is not the case for the \bang{} calculus
~\cite[Sec.2.3]{ehrhard2016bang}. Confluence can be recovered modulo an equivalence relation
contained in the $\sigma$-equivalence generated by the $\sigma$-rules, but
this requires working modulo said equivalence. Another solution to this problem
is to consider the distance variant \emph{(\dbang{})}~\cite{bucciarelli2023bang}, similar to the distant variant of \cbv{}.
This allows us to work without the
$\sigma$-rules and the aforementioned equivalences. This new system has been
shown to be confluent~\cite{bucciarelli2023bang}. It has also been used in
unifying multiple results for \dcbn{} and \dcbv{}. For instance,
in~\cite{arrial-diligence} the authors show that the rewriting results
(confluence, factorization) of \dbang{} carry over to the \dcbn{} and \dcbv{}
setting using the translations into \dbang{}\footnote{Note that while the
translation of \dcbn{} into \dbang{} is the usual one mentioned before, the authors
use a new translations for \cbv{} which will be discussed in
Section~\ref{sec:translations}}. Similarly, the notion of solvability in both \dcbn{} and \dcbv{} has been captured
through the notion of \emph{meaningfulness} in \dbang{}~\cite{kag24,
arrial-genericity}. However, despite significant progress, the approximation theory of the
\dbang{}-calculus remains largely underdeveloped.

A recent work by Mazza \& Dufour tried to close that gap~\cite{mazza-dufour}: they developed a generic notion of Böhm tree and Taylor expansion for a
language called \proc{}, representing untyped proof structures. They showed how
any language that can be embedded into \proc{} in a ``nice way'' inherits the
notions of Böhm tree and Taylor expansion from \proc{}, and their commutation Theorem. In
particular, \dbang{} admit such an embedding. However, no notion of
meaningfulness exists in \proc{} and it is not clear how to relate the results
from \cbn{} and \cbv{} with the notion developed in~\cite{mazza-dufour}.

\subsection{Contributions}

We develop a theory of approximation of the distance $\bang$-calculus
(\dbang{}). To this end we recall the definitions and main results of \dbang{}
from~\cite{kag24,agk24d,bucciarelli2023bang} in Section~\ref{sec:dbang}. We
start by developing the Taylor expansion of \dbang{}
(Section~\ref{subsec:taylor-dbang}) where we introduce the resource calculus
(Section~\ref{subsubsec:dbang-ressources}) and define the approximation relation
(Section~\ref{subsubsec:approx-taylor}) and establish a simulation result
between \dbang{} and its approximants in the Taylor expansion
(Theorem~\ref{thm:simu_dbang_full}). We next develop the Böhm approximants of
\dbang{} (Section~\ref{subsec:bohm-dbang})\footnote{\fnl{In the rest of the
paper, we will call them \emph{approximation trees} instead, due to subtle
difference compared to standard Böhm trees. These differences will be explained
in Section~\ref{subsec:bohm-dbang}.}} and prove a commutation theorem between
the Taylor Expansion of Böhm trees and the Taylor normal form
(Theorem~\ref{thm:taylor-bohm-commute}).

We establish the soundness of our definitions with regard to the standard notion
of Böhm trees and Taylor Expansions in the \fnl{distance call-by-name (\dcbn{})
and distance call-by-value (\dcbv{})} $\lambda$-calculus by translating these
systems into \dbang{} (Section~\ref{sec:translations}), in particular, we show
that the Böhm trees of a term $M$ in \dcbn{} (respectively \dcbv{}) are the same
as its translation into \dbang{} (Theorem~\ref{thm:bohm-trad}). We show a
similar result for Taylor expansion (Lemmas~\ref{lem:taylor_trad_dcbn_dcbv} and
Corollaries~\ref{cor:tradn_tnf-tradv-tnf}).

\fnl{
Finally, we establish the commutation result for both \dcbn{} and \dcbv{} as a
consequence of the \dbang{} commutation
theorem~\ref{commutation-thm-translations}.}

\iflong Proofs are available in the Appendix. \fi

\section{The (distance) Bang-calculus}
\label{sec:dbang}
We begin by recalling the theory of \dbang{}, first
with some results from previous
studies~\cite{kag24,agk24d,bucciarelli2023bang}, before developing its approximation theory via Taylor expansion and approximation trees.


\begin{definition}{(\dbang: terms and contexts)}~
	\label{def:dbang-terms-and-ctxs}

	\begin{tabular}{ll}
		(Terms) & $M,N: = x\mid MN\mid \lambda xM \mid \oc M \mid \der M \mid M\es Nx$ \\
		
		(List contexts) & $L:= \square \mid L[M/x]$ \\

		(Surface contexts) & $S:= \square \mid SM \mid MS \mid \lambda x S \mid
		\der S \mid S[M/x] \mid	M\es Sx$\\

		(Full contexts) &  $F:= \square \mid FM \mid MF \mid \lambda x F \mid \der F \mid F\es Mx \mid
		M\es Fx \mid \oc F$ \\

	\end{tabular}
\end{definition}
Terms include the standard $\lambda$-calculus constructs: variable (ranging over a countably infinite set), application and abstraction. Two additional ones are used to define the $\bang{}$-calculus: the \emph{bang} (or \emph{exponential}) $\oc M$,
representing delayed evaluation of the subterm
$M$, and \emph{dereliction} $\der M$, which reactivates it.
Finally, \emph{explicit substitutions} $M \es N x$ comes from the "at distance" mechanism, and represents 
pending substitutions.
The lambda-abstraction (as usual) and the explicit substitution
bind the variable $x$ in $M$. We use \emph{contexts}, \emph{i.e.} terms with a subterm hole
($\square$) and  we denote by $C\langle M \rangle$ the term
obtained by filling the hole in $C$ by $M$. We distinguish \emph{list, surface} and \emph{full} contexts. List
contexts are sequences of explicit substitutions and will be used for the reduction
at a distance. Surface and full contexts determine whether reduction under a
$\oc$ is allowed (in particular it is forbidden in \emph{weak}
calculi (\cbn{} or \cbv{})).

The \dbang{} calculus has three reduction rules: 
\[
L\<\lambda x M\>N \todb L\<M\es Nx\>
\hfill
M\es{L\<\oc N\>}x \todb L\<M\{N/x\}\>
\hfill
\der{L\<\oc M\>} \todb L\<M\>
\]
Notice that requiring subterms of the form $\oc M$ assigns them \emph{value} status in the \cbv{} sense.

We write $\todbs$ and $\todbf$ for the closures under surface and full contexts, respectively; both are confluent~\cite[Theorem 1]{kag24}.
We also denote by $\to_i$ the internal reductions (\emph{i.e.} $\to_i
= \todbf\backslash\todbs$, these are the reductions occurring under a
$\oc$-construct).	

\begin{example}
	\begin{align*}
	&(\lambda x xy)[\oc z/y]\oc\oc((\lambda ww)\oc N) \\
	&\todbs (xy)[\oc\oc((\lambda ww)\oc N)/x][\oc z/y] \\
	&\todbs \oc((\lambda ww)\oc N)y[\oc z/y] \\
	&\todbs \oc((\lambda ww)\oc N)z \\
	&\todbs \oc(w[\oc N/w])z \\
	&\to_i \ \ \oc Nz
	\end{align*}
\end{example}

\begin{example}~\label{ex:terms}
	\begin{itemize}
		\item $\Delta=\lambda x (x\oc x), \Omega =\Delta\oc\Delta$. We
			have $\Omega\todbs^2 \Omega$.
		\item $Y^n_x=(\lambda y x\oc(y\oc y))\oc (\lambda y x\oc(y\oc
			y))$. We have $Y^n_x\todbs^+ x\oc Y^n_x$
		\item $Y^v_x=(\lambda y x(y\oc y))\oc (\lambda y x(y\oc y))$. We
			have $Y^v_x\todbs^+x Y^v_x$
	\end{itemize}
	The upper scripts in the two last items, as we shall see further,
	represent the fact that $Y^n_x$ and $Y^v_x$ correspond respectively to
	the \cbn{} and \cbv{} translations from a fixpoint combinator $Y$ of standard
	$\lambda$-calculus\footnote{
		Note that $Y^v_x$ lacks some good properties if one would want to use it for
		computations in \cbv{}, and \cbv{} fixpoints are usually defined differently. We
		nevertheless keep $Y^v_x$ as a running example, so as to illustrate an
		infinite computation, with a fixpoint behaviour, but with empty semantics.
	}.
\end{example}

\subsection{Meaningfulness in \dbang}

This section focus on \emph{meaningfulness}, later related to approximation theory in Section~\ref{sec:meaningfulness}: intuitively, a term is meaningful if it reduces to a desired result under some \emph{testing context}.
 It generalizes the notions of \emph{solvability} \cbn{} and \emph{scrutability} \cbv{}, and has been studied in detail for \dbang\
in \cite{kag24}.

\begin{definition}\label{def:dbang_meaningful}
	A term $M$ of $\dbang$ is said \emph{meaningful} if there exists a
	testing context $T:= \square \mid TM \mid (\lambda x T)M$ and a term $P$ such that $T\<M\>\todbs^*\oc P$.
\end{definition}

Notice that the surface reduction involved in meaningfulness is not
restrictive: for any $N$, $N\todbf^*\oc P$, then there is some $P'$
such that $N\todbs^*\oc P'$ (this follows from a standardisation
property, see Corollary~\ref{cor:standardisation_dbang} below). It has been shown in
\cite{kag24} that the smallest theory that
identifies all meaningless terms is consistent; that meaningful terms and surface-normalizing
terms can be characterized by an intersection type
systems and finally that meaningless subterms
do not affect the operational meaning of a given term~\cite[Proposition 8, Theorem 24 and Corollary
11]{kag24}.

A natural property of surface reduction is that it determines the external shape
of a term. In other words, allowing full reduction does not unlock external
redexes. \fnl{Indeed, with terms like \emph{e.g.} $\lambda x\oc M$, $y\oc M$, $\oc M\oc M$,
reductions occurring contextually inside $M$ cannot make the full term reducible
at the surface (the syntax of such term will be precisely characterized
further).} This is expressed as the following factorization  proposition:

\begin{proposition}\label{prop:standardisation_dbang}(\cite{agk24d}, Corollary
	21)
	Let $M\todbf^* N$. There is some $P$ such that $M\todbs^* P\to_i^* N$.
\end{proposition}
Conversely, internal reductions preserve the \emph{external shape}.  We express this notion with multi-holes
surface contexts, that let us reformulate the factorization property as follows:

\begin{corollary}[Standardization]\label{cor:standardisation_dbang}
	Let $S^+$ denote multi-holes surface contexts and given by the syntax:
	$
		S^+:= \square\mid M\mid	S^+S^+\mid S^+[S^+/x]\mid \der{S^+}\mid \lambda xS^+
	$.
	For any reduction $M\todbf^* P$, there are some terms $N_i$ and a
	multi-hole surface context $S$ such that $P=S\<\oc N_1,\dots ,\oc N_k\>$ and
	$
		M\todbs^*S\<\oc N_1',\dots ,\oc N_k'\>\to_i^*S\<\oc N_1,\dots ,\oc N_k\>
	$
\end{corollary}
If $k=0$,  the context has no hole and the reduction occurs only at surface level.


\subsection{Taylor expansion}
\label{subsec:taylor-dbang}

We define the Taylor expansion of \dbang{}, starting with the associated resource calculus \dbangres.
This calculus will also serve as the resource language for \dcbn{} and \dcbv{} in
Section~\ref{sec:translations}.	There is no necessity to define specific resource calculi, as
$\dbangres$ fits well as a target 
of usual Taylor expansion (\cbn{}~\cite{ER08} and
cbv{}~\cite{ehrhard:Collapsing}), with a straightforward
adaptation to distant setting.

\subsubsection{\dbang{}: resources}
\label{subsubsec:dbang-ressources}

\begin{definition}[Resource calculus \dbangres]~
	\label{def:dbang-taylor-terms-and-ctxs}

	\begin{tabular}{ll}
	(terms) & $m,n: = x\mid mn\mid \lambda xm \mid  \der m \mid m\es nx\mid
	[m_1,\dots ,m_k]$ \\

	(lists) & $l:= \square \mid l[m/x]$ \\

	(surface) & $s:= \square \mid sm \mid ms \mid \lambda x s \mid
	\der s \mid s[m/x] \mid	m\es sx$\\
	
	(full) & $f:= \square \mid fm \mid mf \mid \lambda x f \mid
	\der f \mid f[m/x] \mid	m\es fx\mid [f,m_1,\dots ,m_k]$\\

	(tests)& $t:= \square \mid tm \mid (\lambda x t)m$
	\end{tabular}
\end{definition}

\paragraph*{Notations (resource terms, multisets/bags, multilinear substitutions)}
Terms $[m_1,\dots ,m_k]$ for $k\in\mathbb N$, called \emph{bags}, denote finite
multisets of resource terms, with $[~]$ the empty bag. When necessary, we use a subscript as $[m]_k$ or
$[m,\dots ,m]_k$ in order to make explicit the number of occurrences.

Contexts are exactly
as in \dbang{}
(Definition~\ref{def:dbang-terms-and-ctxs}), with bags replacing exponentials.
We write $P_k$ for the sets of
permutations of the set $\{1,\dots ,k\}$, and we denote as $d_x(m)$ the number of
free occurrences of the variable $x$ in $m$. Resource substitution is
(multi-)linear, \fnl{and ordered}: when we write
$m\{n_1/x_1,\dots ,n_k/x_k\}$, it is always intended (in the resource setting) 
that $x_i$ represents the $i$-th free occurrence of $x$ in $m$. In that way,
each term $n_i$ is substituted exactly once in $m$.

We are now ready to define the reduction relation.
\begin{definition}~
	\label{def:dbang-taylor-reductions}
	
	The reduction relation $\Todbr\subseteq\dbangres\times\mathbb
	\wp(\dbangres)$ is
	then defined as follows:
	\begin{itemize}
		\item $l\<\lambda x m\>n \Todbr \{l\<m\es nx\>\}$
		\item $m\es{l\<[n_1,\dots ,n_k]\>}x \Todbr 
		\left\{\begin{array}{ll}
			\displaystyle{\bigcup_{\sigma\in
			P_k}}l\<m\{n_{\sigma(1)}/x_1,\dots ,n_{\sigma(k)}/x_k\}\> & \text{ if }
			k=d_x(m) \\
			\emptyset & \text{ otherwise}
		\end{array}\right.$
		\item $\der{l\<[m_1,\dots ,m_k]\>} \Todbr \{l\<m_1\>\}$ if $k=1$ and
			$\emptyset$ otherwise
	\end{itemize}
\end{definition}

We write $m\todbr n$ as soon as $m\Todbr X$ and $n\in X$ for some $n$ (if
$m\Todbr \emptyset$, we also abusively write $m\todbr \emptyset$, and we add the
following equation: if $f\<\emptyset\>=\emptyset$ for any full context $f$). We
also define $\todbrs$ and $\todbrf$ the contextual closures of $\todbr$ under
surface and full contexts, respectively. Notice that none of these reductions is
deterministic.
Both $\todbrs$ and $\todbrf$ are strongly normalizing, which is an immediate
consequence of linearity: the size of bags of resource terms is decreasing. Confluence of $\todbrs$ and $\todbrf$ follows from standard resource-calculus arguments and from the proofs for
$\todbs$ and $\todbf$. 

\subsubsection{Approximation}
\label{subsubsec:approx-taylor}

In Figure~\ref{fig:approx_dbang}, defines the relation
$\app\subseteq\dbangres\times\dbang$, where $m\app M$ means that
$m$ is a multilinear resource approximation on $M$.

\begin{figure}
\begin{center}
	\AxiomC{}\UnaryInfC{$x\app x$}\DisplayProof
	\qquad
	\AxiomC{$m\app M$}\UnaryInfC{$\lambda xm\app\lambda xM$}\DisplayProof
	\qquad
	\AxiomC{$m\app M$}\UnaryInfC{$\der m\app\der M$}\DisplayProof
	\qquad
	\AxiomC{$m\app M$}\AxiomC{$n\app N$}\BinaryInfC{$mn\app MN$}\DisplayProof
	
	\vspace{1em}
	\AxiomC{$m\app M$}\AxiomC{$n\app N$}\BinaryInfC{$m[n/x]\app M[N/x]$}\DisplayProof
	\qquad
	\AxiomC{$m_1\app M$}\AxiomC{\dots }\AxiomC{$m_k\app
	M$}\RightLabel{$k\in\mathbb N$}\TrinaryInfC{$[m_1,\dots ,m_k]\app \oc M$}\DisplayProof
	\caption{Resource approximation for \dbang}
	\label{fig:approx_dbang}
\end{center}
\end{figure}

We extend this definition to list contexts as follows: 
$\square\app\square$, $l[m/x]\app L[M/x]$ if $l\app L$ and $m\app M$. The
extension to surface and full contexts follows analogously.

The expected behaviour of context approximation is given by the following lemma (proved by induction on contexts):

\begin{lemma}~
	\label{lem:approx_contexts}
	\begin{itemize}
		\item If $m\app L\<N\>$, then there exist $l\app L$ and $n\app
			N$ such that $m=l\<n\>$.
		\item If $m\app S\<N\>$, then there exist $s\app S$ and $n\app
			N$ such that $m=s\<n\>$.
	\end{itemize}
\end{lemma}

Notice that this property fails for full contexts. \fnl{Indeed, consider
	$t=[]\app\oc M=F\<M\>$, there is no $m\app M$ or $f\app F$ such that $m\app M$ and
	$t=f\<m\>$.
}
Consequently, our
definitions do not provide a convenient notion of approximation between
full contexts. This is due to a need for parallel treatment of bags, which - as we shall see later - is incompatible with single-hole contexts.

\begin{definition}{(Taylor expansion)}~
	\label{def:dbang-taylor-expansion}
	For any $M\in\dbang$, we define its Taylor expansion as the set of
	its resource approximants:
	\[
		\taylor M=\{m\in\dbangres\mid m\app M\}
	\]
\end{definition}

By strong normalization of $\dbangres$, we can define the normal form $\nf m$ of a resource term $m$ as the finite set made of its full
reducts. \fnl{Remember that the resource reduction we consider is not
	deterministic ($\Todbr$ is, but not $\todbrf$), hence this
set can contain multiple terms, and can be empty.} We then define the \emph{Taylor normal
form} of \dbang{} terms as $\taylornf M =\bigcup_{m\app M}\nf m$. Notice that
Taylor normal form is made of full normal terms, not only surface normal forms.

\begin{remark}
	\fnl{If $\taylor M\Todbr^* X$ and $X$ is a set of full normal terms, than
		$X=\taylornf M$. But the converse is not true: we can have
		$\taylornf M$ non empty, but no $k$ such that $M\Todbr^k
		\taylornf M$. This is due to the fact that terms in $\taylor M$
		can require different number of reduction steps, as with Taylor
		expansion of fixpoint combinators where this number is unbounded,
		as illustrated in Example~\ref{ex:terms_app}.
	}
\end{remark}

\begin{example}\label{ex:terms_app} Consider the terms given in Example~\ref{ex:terms}.
	\begin{itemize}
		\item An approximant of $m\app\Omega$ must be of shape $(\lambda x
			x[x]_k)[\lambda x x[x]_{k_1},\dots ,\lambda x x[x]_{k_l}]$.
			Now, if $k=l-1$ (otherwise $m\todb \emptyset$) then $m\todbrs (\lambda x
			x[x]_{k_1})[\lambda x x[x]_{k_2},\dots ,\lambda x
			x[x]_{k_l}]$\footnote{
				We consider here one possible reduction, any
				element of the bag could be substituted to the
				inner head variable $x$, not necessarily $\lambda
				x[x]_{k_1}$, the argument is valid for all the
				reduction paths.
			}, which is again an approximant of $\Omega$. But we can
			observe that the cardinality of the bag reduces during
			this reduction; hence if we iterate this reduction, we
			eventually reach a term like $(\lambda
			x[x]_k)[]\todbrs\emptyset$ (if an empty reduction has
			not occurred before). So, $\taylornf\Omega=\emptyset$.
		\item Similarly, if $m\app Y^n_x$, we verify easily that
			$m\todbrs x[n_1,\dots ,n_k]$, with $n_i\app Y^n_x$. In
			particular, $x[]\app Y^n_x$ and is in normal form.
			Actually, $\taylornf{Y^n_x}$ can be characterized
			inductively : $x[]\in\taylornf{Y^n_x}$, and if
			$n_1,\dots ,n_k\in \taylornf{Y^n_x}$, then
			$x[n_1,\dots ,n_k]\in\taylornf{Y^n_x}$, for any $k$.
		\item The other fixpoint term, $Y^v_x$, behaves slightly
			differently: if $m\app Y^v_x$, then we verify
			$m\todbrs^* xn$, for some $n\app Y^v_x$. But here,
			because of the argument not being in a bag, all
			approximants reduce (if not $\emptyset$) to some term
			$xxx\dots  xn$, but are not in normal form; since such a
			reduction terminates, we observe that
			$\taylornf{Y^v_x}=\emptyset$.
	\end{itemize}

\end{example}
\begin{remark}\label{rem:clashes}[Clashes and normal forms]
	Note that we have not excluded so-called \emph{clashes} from our calculus,
	as they cause no issue in our setting. Thus, terms such $\der{\lambda xx}$ are considered as
	regular normal forms. This approach is similar to the one of Dufour and
	Mazza~\cite{mazza-dufour}. However, giving an empty semantics to clashes could
	easily be done, for example, adding reductions such as $\der{\lambda
	xm}\to\emptyset$ to resource calculus.
\end{remark}

\subsubsection{Simulation}\
We now formalize the simulation of $\dbang$ reduction in the approximants of $\dbangres$, as sketched in the
Example~\ref{ex:terms_app}.

\begin{lemma}[Substitution Lemma for Taylor expansion]
	\label{lem:taylor:substitution-lemma}
	For any $M,N$ of $\dbang$, for any $m,n_1,\dots ,n_{\deg xm}$ of $\dbangres$,
	we have $m\{n_1/x_1\dots n_{\deg xm}/x_{\deg xm}\}\app M\{N/x\}$ if and only
	if $m\app M$, and $n_i\app N_i$ for all $i\leq\deg xm$.
\end{lemma}

We can show, by a straightforward induction, with the help of the Substitution
Lemma stated above, that the surface reduction acts
exactly the same way in a \dbang{} term and in its approximants: 

\begin{lemma}\label{lem:simu_dbr_n}
	Let $m\app M$ and $m\todbrs n$. Then there is
	$N\in\dbang$ such that $M\todbs N$ and $n\app N$.
\end{lemma}

However, this result is false when considering reduction in full
contexts: take the following term $m=[(\lambda xx)y,(\lambda xx)y]$.
We have $m\app\oc ((\lambda xx)y)$ and $m\todbrf n= [x[y/x],(\lambda
xx)y]$, but $n$ is not an approximant of any $\dbang$ term (\fnl{it has lost its
uniformity, in the sense of Ehrhard and Regnier~\cite{ER08}}). We will
need parallel reduction to achieve full simulation, this is done in
Subsection~\ref{subsec:parallel}.

\begin{lemma}
\label{lem:square}
Considering the following configuration:
\begin{center}
	\begin{tikzpicture}
		\node(M)at(0,0){$M$};
		\node(N)at(2,0){$N$};
		\node(m)at(0,-1){$m$};
		\node(n)at(2,-1){$n$};

		\draw(1,0)node{$\todbs$};
		\draw(1,-1)node{$\todbrs$};
	
		\draw (0,-0.5)node{$\triangledown$};
		\draw (2,-0.5)node{$\triangledown$};
	\end{tikzpicture}
\end{center}
Given three of four terms,
we can always obtain a convenient fourth that makes the square commute.
\iflong (for a proof, see lemmas~\ref{lem:simu_dbr_one}, \ref{lem:simu_dbr} and \ref{lem:simu_dbr_bis} in
 Appendix~\ref{appendix:dbang}).
\fi
\end{lemma}

\begin{definition} 
	Let $X$ and $Y$ be sets, and $\mathcal R\subseteq X\times Y$ a relation. We write
	$X\extend{\mathcal R} Y$ whenever $\forall x\in X$, $\exists y\in Y$ such that $(x,y)\in \mathcal R$
	and $\forall y\in Y$, $\exists x\in X$ such that $(x,y)\in \mathcal R$.
\end{definition}

Then, using Lemma~\ref{lem:square} we obtain a simulation theorem for
surface reduction:

\begin{theorem}[Simulation]
	\label{thm:simu_taylor_db}
	Let $M,N$ in \dbang{} such that $M\todbs N$. We have $\taylor M
	\extend{\todbrs} \taylor N$.
\end{theorem}

\subsubsection{Parallel reduction and full contexts}
\label{subsec:parallel}

Figure~\ref{fig:parallel_reduction} defines a parallel resource reduction allowing us to handle full contexts. Intuitively, it needs to
reduce at once every term occurring in a bag, in order to simulate internal 
reductions like $\oc M\todbf \oc M'$: \fnl{in this case, approximants of the antecedent are of
the form $[m_1,\dots ,m_k]$ with $m_i\app M$, and the simulation needs to reduce all the
$m_i$'s at once, in order to obtain an approximation of the reduct $[m'_1,\dots ,m'_k]\app \oc
M'$.} In particular, the reduction needs to be reflexive
because \emph{e.g.}
$[]\app \oc M$ and $[]$ does not reduce to $\emptyset$, we have to consider that $[]\app
\oc M'$ is obtained from $[]$ by reduction. The resulting relation $\totodbr$, is a
well-known non-deterministic extension of standard reduction which can be used
to prove confluence property. In general, for a reduction $\to$, we have
$\to\subsetneq\rightrightarrows\subsetneq\to^*$\footnote{
	\fnl{Usually the parallel reduction enjoys the diamond property. It is
		not the case here because of the non-determinism induced by
		the permutations; in order to recover it, it is necessary to
		lift the reduction to set of terms. We mentioned it with non parallel
		reduction, whith $\Todbr$ being confluent, but not $\todbrs$,
		for the same reason.
	}
}. See \emph{e.g.} Barendregt's proof of confluence
for the $\lambda$-calculus~\cite{barendregt}, known as Tait—Martin-Löf technique.
\footnote{\fnl{See
Vaux-Auclair~\cite{Vaux19} for a detailed study of this reduction in $\lambda$-calculus and its
interaction with quantitative Taylor expansion.}}
\begin{figure}
\begin{center}
	\AxiomC{}\UnaryInfC{$x\totodbr x$}\DisplayProof
		\qquad
	\AxiomC{$m_1\totodbr m_1'$}\AxiomC{\dots }\AxiomC{$m_k\totodbr
	m_k'$}\RightLabel{$k\in\mathbb N$}\TrinaryInfC{$[m_1,\dots ,m_k]\totodbr
	[m_1',\dots ,m_k']$}\DisplayProof
	
		\vspace{1em}
	\AxiomC{$m\totodbr m'$}
	\RightLabel{*}
	\UnaryInfC{$s\<m\>\totodbr s\<m'\>$}\DisplayProof
		\qquad
	\AxiomC{$[m_1,\dots ,m_k]\totodbr [m_1',\dots ,m_k']$}
	\AxiomC{$n\totodbr n'$}
	\AxiomC{$l\totodbr l'$}
	\RightLabel{**}
	\TrinaryInfC{$n[l\<[m_1,\dots ,m_k]\>]\totodbr
	l'\<n'\{m'_{\sigma(1)}/x_1,\dots ,m'_{\sigma(k)}/x_k\}\>$}\DisplayProof
		
		\vspace{1em}
	\AxiomC{$m\totodbr m'$}
	\AxiomC{$n\totodbr n'$}
	\AxiomC{$l\totodbr l'$}
	\TrinaryInfC{$l\<\lambda x m\> n\totodbr l'\<m'[n'/x]\>$}
	\DisplayProof
		\qquad
	\AxiomC{$m_1\totodbr m_1'$}
	\RightLabel{***}
	\UnaryInfC{$\der{[m_1,\dots ,m_k]}\totodbr m_1'$}
	\DisplayProof

	\begin{itemize}
		\item[*] $s$ is any surface resource context
		\item[**] $\sigma\in P_k$, and if $k=d_x(n')$ (otherwise the
			reduction gives $\emptyset$).
		\item[***] $k=1$. Otherwise, the reduction gives $\emptyset$.
	\end{itemize}
	\caption{Parallel reduction for \dbangres}
	\label{fig:parallel_reduction}
\end{center}
\end{figure}
\iflong We abusively write $l\totodbr l'$ for contexts as soon as
$l=\square[m_1/x_1]\dots [m_k/x_k], l'=\square[m'_1/x_1]\dots [m'_k/x_k]$ and $m_i\totodbr
m'_i$ for any $i\leq k$. \fi

Since $\totodbr$ is size-decreasing, it  enjoys weak normalization,  but
obviously not strong, as it is reflexive. We can now state our simulation
results for full contexts:

\begin{lemma}\label{lem:simu_dbr_full}
	Let $M,N\in\dbang$ with $M\todbf N$. For any $m\app M$, either
	$m\totodbr \emptyset$ or there is some $n\app N$ such that $m\totodbr n$
\end{lemma}
The symmetric counterpart of this result is obtained with a similar reasoning:
\begin{lemma}\label{lem:simu_dbr_full_bis}
	Let $M,N\in\dbang$ with $M\todbf N$. For any $n\app N$,
	there is some $m\app M$ such that $m\totodbr n$.
\end{lemma}

The two previous lemmas give us the desired simulation result for full reduction:
\begin{theorem}\label{thm:simu_dbang_full}
	Let $M,N\in\dbang$ such that $M\todbf N$, then we have
	$
		\taylor M \extend{\totodbr}\taylor N
	$
\end{theorem}

We saw that surface reduction acts similarly in $\dbang$ and in $\dbangres$,
while parallel reduction is necessary to give a multilinear account to internal
reductions (the definition of parallel reduction alone does not imply that they
occur exclusively inside bags, but this is made mandatory through the use of invariant multi-hole
surface contexts). The factorization properties established for \dbang{}
(Corollary~\ref{cor:standardisation_dbang}) can easily be
translated in the resource setting:

\begin{proposition}[Factorization]\label{prop:standardisation_resources}

	Let $s^+$ denote multi-holes surface contexts and given by the syntax
	$
		s^+:= \square\mid m\mid	s^+s^+\mid s^+[s^+/x]\mid \der{s^+}\mid \lambda xs^+
	$.
	For any reduction $m\todbrf^* p$, there are some bags $\bar n_i$ and some
	multi-hole context $s$ such that $p=s\<\bar n_1,\dots ,\bar n_k\>$ and 
	$
		m\todbrs^*s\<\bar n_1',\dots ,\bar n_k'\>\totodbr^*s\<\bar n_1,\dots ,\bar n_k\>
	$
\end{proposition}


Following the reduction occurring on resource terms, we have defined
previously Taylor normal form. \iflong We develop in this section some lemmas
on those objects. They will \else The following lemmas will\fi be useful in order to prove the
Commutation Theorem between approximation trees and Taylor expansions. 
\begin{lemma}
	\label{lem:bohm-taylor:equiv-imply-same-nfTaylor}
	Given $M \todbf^* N$ then $\taylornf M = \taylornf N$.
\end{lemma}

\begin{lemma}
	\label{lem:taylor:approx-in-reduces-bis}
	Given $m \in \taylornf M$ then there exists $M'$ such that $M \todbf^*
	M'$ and $m\app {M'}$.
\end{lemma}

\subsection{Approximation trees}\label{subsec:bohm-dbang}

A classical way to approximate $\lambda$-terms is through evaluation trees, which can be seen as infinitary normal forms. 
The evaluation tree of a $\lambda$-term corresponds to the limit of a sequence
of increasingly precise approximants of this $\lambda$-term. Those approximants
are terms in normal forms, but using potentially a specific subterm $\bot$,
representing an undefined computation. When reducing the $\lambda$-term some
subterms become stable under further reductions, these subterms are precisely
those captured by approximants. 

Different comparable evaluation trees coexist. The original and best-known ones are Böhm trees. They have been introduced in the \cbn{} paradigm where they intersect with numerous notions. Indeed, they give a precise classification of $\lambda$-terms in regard of their operationnal behaviour.

In \cbn{}, using all partial normal forms (possibly containing $\bot$) as approximants, corresponds to Lévy-Longo trees \cite{Longo83}. Böhm trees use a restriction on this set of approximants, so that unsolvable terms have an empty Böhm tree. This restriction is based on restricting approximants to (partial) head normal forms, since being solvable is equivalent to having a head normal form. Lévy-Longo trees instead rely on weak head normal forms. In fact, it only leads to the suppression of "stupid" approximants such as $\lambda x\bot$. Böhm trees and Lévy-Longo trees differ essentially in their treatment of abstractions over undefined computations and therefore share similar properties. Among others, they both satisfy variants of the Commutation Theorem, which states that Taylor expansion of the Böhm tree of a term is equal to its Taylor normal form ~\cite{ER08}.

More recently, Böhm trees have been introduced in the \cbv{} setting (in a variant with permutation rules)~\cite{kerinec:bohmTreeCbV}. 
They are defined without restriction on approximants. Indeed, in this case,
there is no known syntactic characterization of solvability in terms of a
specific normal form. In that, they may be considered closer to the Lévy-Longo
trees. However, the relation with scrutability — the notion of meaning — in
\cbv{} is respected: unscrutable $\lambda$-terms have an empty Böhm tree. An
approximant such as $\lambda x\bot$ is meaningful in this setting, since it is a
value.

In this work we define approximation trees for the distant version of \cbv{}, \cbn{} and \bang. As in the \cbv{} case, we do not restrict our approximants for \dcbn{} and \dcbv{}. In order to be consistent between the different calculus, we adopt the same approach for \dcbn{}. The resulting trees are called approximation trees here. They correspond to Lévy-Longo trees in the \dcbn{} case, while in the \dcbv{} case they are closer to Böhm trees in spirit, although the situation is more subtle. Since those two family of trees are very strongly related, we choose not to focus too much on whether our trees should be considered as Böhm or Lévy–Longo trees.\\

In this section we will introduce those trees for dBang and their relation to the Taylor expansion via the Commutation Theorem (Theorem~\ref{thm:taylor-bohm-commute}).

\begin{definition}[$\dbangb$]
	Let $\dbangb$ be the set of $\dbang$-terms extended with the symbol $\bot$.
	In the following we use subscripts as $\bot_i$ in order to distinguish
	occurrences of $\bot$ as a subterm.
	Similarly, we extend the different types of contexts; and
	we extend reductions in the obvious way.
\end{definition}
\begin{definition}\label{def:bohm_approx}
	The set of approximants is a strict subset of $\dbangb$  generated by the grammar:
\begin{minipage}{0.49\textwidth}
	$$\begin{array}{lll}
	A&:= & \bot\mid B\mid \lambda xA\mid \oc A\mid A\es{A_\oc }x\\
	B&:=& x\mid A_\lambda A\mid \der{A_\oc}\\
	\end{array}$$
\end{minipage}
\begin{minipage}{0.49\textwidth}
	$$\begin{array}{lll}
	A_\oc &:= & B \mid \lambda x A \mid A_\oc \es{A_\oc}x\\
	A_\lambda &:= &B \mid \oc A\mid A_\lambda \es{A_\oc}x\\
	\end{array}$$
\end{minipage}
\end{definition}

\begin{lemma}
	Approximants are the normal forms of $\dbangb$. 
\end{lemma}
Substituting $\bot$ by any term preserves normal forms, as a consequence of the syntactic structure of
approximants. First, approximants contain no redexes. Secondly, 
there are no approximants containing
subterms of shape $\der \bot, \bot A, A\es\bot x$ that could hide a
potential redex. Formally:
\begin{lemma}
	\label{lem:approx_bot}
	Let $A$ be an approximant, $M$ any term of $\dbangb$, and $\bot_i$ some
	occurrence of $\bot$ in $A$. If	$A[M/\bot_i]\todbf N$, then
	$N=A[M'/\bot_i]$ with $M\todbf M'$. \emph{i.e.} we cannot create a redex
	when replacing a $\bot$ by a term $M$ in an approximant: the only redexes
	obtained this way are those already present in $M$.
\end{lemma}

\begin{definition}[Order]
Let $\smaller\ \subseteq \dbangb\times \dbangb$ be the least contextual closed
	preorder on $\dbangb$ generated by setting: 
 $\forall M\in \dbangb$ and for all full context $F$, $F[\bot] \subseteq F[M]$. 
\end{definition}
In other words, $M\smaller N$ if and only if $M=N\overrightarrow{\{P/\bot\}}$. In
the evaluation trees literature, we sometimes find approximation orders where only
terms of shape $\oc P$ are replaced by a $\bot$ (see \emph{e.g.} Dufour and
Mazza~\cite{mazza-dufour}). For technical reasons (Theorem~\ref{thm:bohm-trad} in
particular) due to \cbn{} and
cbv{} embeddings into $\dbang$, we need, for example, $\lambda
x\bot\smaller\lambda xM$ to be a valid approximation.

The following lemma states that when a reduction occurs in a term,
it cannot be \emph{seen} by its approximations, that only represent a part of
the skeleton of their normal form. In other words, when $A\smaller M$, then
$A$ represents a subtree of $M$ which cannot be modified by some reduction.

\begin{lemma}
	\label{lem:bang:bohm:reduce-have-at-least-same-approxs}
	$A\smaller M$ and $M \todbf^* N$ then $A\smaller N$.
\end{lemma}

\begin{definition}[Set of approximants]
\label{def:bang:bohm:approx-set}
Given $M \in \dbang$, the set of approximants of $M$ is defined as follows:
$\setApprox(M) = \set{A \mid M \todbf^* N, A \smaller N}$.
\end{definition}

Note that $\setApprox{(M)}$ is never empty, since it contains at least $\bot$.

\begin{example}\label{ex:terms-approx}
	Consider the terms given in Example~\ref{ex:terms}.
	\begin{itemize}
		\item The only reducts of $\Omega$ being $\Omega$ itself, and
			since the syntax $A$ of approximants contains no redex,
			we easily conclude that $\setapprox{(\Omega)}=\{\bot\}$.
		\item The reducts of $Y_x^n$ are of shape $x!x!x\dots !Y_x^n$, we
			conclude that $\setapprox{(Y_x^n)}$ contains precisely
			the terms $\bot,x\bot,x\oc\bot,x\oc x\bot,x\oc
			x\oc\bot,\dots $.
		\item The reducts of $Y_x^v$ follow the same observation, minus
			the exponentials, and $\setapprox{(Y_x^v)}$ contains
			precisely $\bot, x\bot,xx\bot,xxx\bot,\dots $
	\end{itemize}
\end{example}

\begin{lemma}
	\label{lem:bang:bohm:reduce-have-same-approxs}
	$M \todbf^* N$ then $\setApprox(M) = \setApprox(N)$.
\end{lemma}

\begin{remark}~\label{rem:approx-normal}
	If $A=F[\vec\bot]\in\setapprox(M)$ (where $F$ is a
	multi-hole context), then $M\to^*F[\vec N]$ for some $\vec N$. In particular, if
	$A$ contains no $\bot$ subterms, then $M$ has a normal form, which is
	exactly $A$.
\end{remark}

\begin{definition}[Ideal]
	Let $X\subseteq \dbangb$. We set $X$ is an \emph{ideal} when $X$ is
	downwards closed for $\smaller$, and directed: for all $M, N \in X$
	there exists some upper	bound (with respect to $\smaller$) to
	$\{M,N\}$.
\end{definition}

\begin{lemma}\label{lem:ideal}
	Let $M\in\dbang$. $\setApprox(M)$ is an ideal.
\end{lemma}
	
In order to define approximation trees, we follow the long-established tradition for term
rewriting systems, including \cbn{} and \cbv{} \cite{Ariola70, Boudol83,Barendregt84,
Amadio98, Ariola02, kerinec:bohmTreeCbV}, based on the ideal completion. 
This method constructs the set of ideals of approximants (ordered by a
direct partial order). Evaluation trees, such as Böhm trees are then identified with such ideals.
The finite and infinite ideals represent respectively the finite and
infinite trees. For a $\lambda$-term $M$, its evaluation tree is the ideal
generated by its set of approximants; equivalently, it can be seen as
the supremum of these approximants in the associated directed-complete
domain.
This domain admits a concrete presentation as a coinductive grammar extending
that of approximants, where constructors may be unfolded infinitely often.
Given $M \in \dbang$, $\setApprox(M)$ has a supremum, noted $\cup\setApprox(M)$,
which is a potentially infinite tree.
\begin{definition}[Approximation Tree]
	The approximation tree of a term $M$ in $\dbang$, is given by
	$\cup\setApprox(M)$, and denoted $\btb(M)$.
\end{definition}

Approximation trees satisfy the following immediate properties (we already used these facts on approximants for previous results, they lift immediately to approximation trees).
\begin{proposition}[Some properties of approximation trees]\label{prop:bohm-properties}~
	\begin{itemize}
		\item If $M$ is in normal form, then $\btb(M)=M$.
		\item $\btb(M)=\bot$ if and only if
			$\setApprox(M)=\{\bot\}$\footnote{
				This distinguishes our approximation trees from those of
				Mazza and Dufour~\cite{mazza-dufour}, in which
				$\btb(M)=\bot$ as soon as $M$ has no surface
				normal form. Indeed, for technical reasons
				(relative to \cbv{}
				embeddings), if $\btb(M)=\bot$,
				we need to have $\btb(\lambda xM)=\lambda
				x\bot$, and not $\bot$, because $\lambda x\bot$
				embeds to $\oc(\lambda x\bot)$, while $\bot$
				embeds to $\bot$, which loses the exponential
				and breaks commutation properties between approximation
				trees and \cbv{} embedding
				(Theorem~\ref{thm:bohm-trad}). This distinction
				vanishes at the semantical level: as soon as we
				consider Taylor expansion of approximation trees, the
				trees having no surface normal form are given an
				empty expansion
				(Definition~\ref{def:taylor-dbangb}).
			}, if and only if any reduct of $M$ is a redex.
		\item If $M\to^*N$, $\btb(M)=\btb(N)$, for example
		 $\btb(\der{\oc M})=\btb(M)$.
		\item $\btb(\lambda xM)=\lambda x\btb(M)$ and $\btb(\oc M)=\oc(\btb(M))$.
		\item If $\btb(M)\neq L\<\lambda x-\>$, then
			$\btb(MN)=\btb(M)\btb(N)$.
		\item If $\btb(M)\neq L\<\oc-\>$, then $\btb(\der M)=\der{\btb(M)}$ and
			$\btb(N[M/x])=\btb(N)[\btb(M)/x]$
	\end{itemize}
\end{proposition}
In particular, we can infer from the facts above that if $\btb(M)$ is an application, then
it is equal to $(x)\btb(N)$ for some $N$ (hence $M\to^*L\<x\>N$).

\begin{example}\label{ex:terms-bt}
	From Example~\ref{ex:terms-approx}, we can infer that
	$\btb(\Omega)=\bot$, $\btb(Y^n_x)$ is the infinite application $x\oc
	x\oc x\oc x\dots $, and $\btb(Y^v_x)$ is also an infinite application
	$xxxxxx\dots $. This is the intended behaviour of approximation trees, as in this
	category of terms they represent, at the limit, the amount of result
	produced by a computation, even if the term itself has no normal form.
\end{example}

\subsection{The commutation between approximation tree and Taylor expansion}
\label{subsec:bang:commutation}
Combining the previous results, we state the Commutation Theorem. We start by defining the Taylor expansion of approximation trees.

\begin{definition}\label{def:taylor-dbangb}
We extend the Taylor expansion to $\dbangb$-terms by setting
	$\taylor \bot=\emptyset$ (recall that $f[\emptyset]=\emptyset$ for any
	full context $f$). In other words, there is no $\dbangres$ term $m$ such
	that $m\app \bot$.
\end{definition}

One immediate consequence of this definition is that, for $M$ in $\dbangb$, $\taylor
M\neq\emptyset$ if $M$ contains no $\bot$ as a surface subterm. In other words, $M$
expands to a non-empty set if and only if the $\bot$ are under exponentials
$\oc$ (in that case, these exponentials can be approximated by empty bags).
Moreover, $M\smaller N$ implies $\taylor M\subseteq \taylor N$.

\begin{definition}[Taylor expansion of approximation trees]\label{def:taylor_bohm}
	Given $M\in \dbang$, $\taylor{\btb(M)}=\cup_{a\in\setapprox(M)}\taylor a$.
\end{definition}

\begin{example}\label{ex:bt-taylor}
	Following the terms of previous examples, we easily check that
	$\taylor{\btb(\Omega)}=\emptyset$ (because
	$\setapprox{(\Omega)}=\{\bot\}$).
	Recall that $\setapprox{(Y^n_x)}=\{\bot,x\bot,x\oc\bot,\dots \}$. The first of
	these approximants having a non-empty expansion is $x\oc\bot$,
	approximated by $x[]$, 
	$\taylor{\btb(Y^n_x)}=\{x[],x[x[],\dots ,x[\dots ]],\dots \}$.
	And
	$\taylor{\btb(Y^v_x)}=\emptyset$. Indeed, all approximants of $Y^v_x$ are
	of shape $xxx\dots x\bot$, and have a surface $\bot$ that expands to
	$\emptyset$.
	With examples~\ref{ex:terms-approx}
	and~\ref{ex:terms_app}, we observe for these terms the identity of $\taylornf{M}$ and
	$\taylor{\btb(M)}$ which is proved globally in this section.
\end{example}
Remark that, since $\setapprox(M)$ is an ideal, $\taylor{\btb (M)}$ is a directed
union. The purpose of Taylor expansion is to approach a term by finitary resource
terms only, so we do not consider any infinite supremum of this union, and keep
a set of terms, this approximation being inductive; while Böhm trees are
coinductive and consist of infinite objects.

\fnl{The two following lemmas formalize a compatibility between approximation order
and inclusion of Taylor expansion.}

\begin{lemma}\label{lem:approx-taylor-subset}
	Let $A\smaller M$, then $\taylor A\subseteq\taylor M$.
\end{lemma}

\begin{lemma}\label{lem:approx_subset_nft}
	Let $A\in\setapprox(M)$, then $\taylor A\subseteq\taylornf M$.
\end{lemma}

\fnl{Analogously, normalized resource approximants match the tree approximants.}

\begin{lemma}
	\label{lem:taylor-bohm:nf-imply-in-approx}
	Let $m\app M$ in normal form, then there exists an approximant $A$ such that
	$A\smaller M$ and $m \app A$
	(remember that $m\app M$ and $m\in\taylor M$ are the same thing).
\end{lemma}

From the previous lemmas we deduce the Commutation Theorem:
\begin{theorem}
	\label{thm:taylor-bohm-commute}
	Let $M\in\dbang$. $\taylor{ \btb(M)} =\taylornf M$.
\end{theorem}

\section{Translations}
\label{sec:translations}
This section is about the approximation theory of \dcbn{} and \dcbv{}, in particular about how the
embeddings into $\dbang$ can profit from the result of previous Section. 
We briefly recall the relevant definitions and results of those embeddings, that have been already well studied  ~\cite{bucciarelli2023bang,kag24,arrial-quantitative,arrial-diligence,agk24d,arrial-genericity}
$\dcbv$ and $\dcbn$ share the same syntax:
$M, N ::= x \alt \lambda x M \alt MN \alt M\es Nx$.
%
Furthermore, the values $V$ are defined as either a variable $x$ or a
$\lambda$-abstraction $\lambda x M$. List contexts are defined as in \dbang{}; reduction rules are:
\begin{itemize}
	\item In $\dcbn$: $L\langle \lambda x M \rangle~N \ton L\langle M\es
		Nx\rangle$ and $M\es Nx \ton M\s Nx$
	\item In $\dcbv$:$L\langle \lambda x M \rangle~N \tov L\langle M\es
		Nx\rangle$ and $M\es {L\<V\>}x \tov L\<M\s Nx\>$
\end{itemize}
We are now ready to define the translations of \dcbn{} (noted $()^n$) and 
\dcbv{} (noted $()^v$) into \dbang. Several translations of \dcbv{} into \dbang{} exist (or the original \bang-calculus
without explicit substitutions). The first one~\cite{ehrhard2016bang}, inspired
by Girard \emph{second translation}, does not
preserve normal forms ($xy$ translates to $\der{\oc x}\oc
y$). Another translation was then proposed~\cite{bucciarelli2023bang}, which
fixes this problem by simplifying the created
redexes by the translations on the fly. This is the translation that we use
here. It is worth noticing that this translation does not satisfy reverse
simulation from \dbang{} to \dcbv{}~\cite[Figure 1]{agk24d}
This issue has been addressed by a third translation~\cite{agk24d}.
However, reverse simulation is not necessary in our case, and we will
keep with the translation proposed in~\cite{bucciarelli2023bang},
although we are convinced that the same developments could be carried
out with the translation from~\cite{agk24d}. This choice is motivated
by the fact that the translation we consider fits better with the
Linear Logic discipline from which stems Taylor expansion, and also
because when considering the study of meaningfulness, we can rely on
results that have been proved for this translation. Moreover, this
translation enjoys a weaker property than strict reverse simulation,
that we call embedding (see Lemma~\ref{lem:trad_embedding}) and which
is sufficient for our study (approximation trees and Taylor expansion concern
mostly iterated reductions $\to^*$, and one-step, thus reverse
simulation is then not mandatory for our results).
\begin{center}
\begin{minipage}{0.30\textwidth}
\begin{align*}
	\tradn x &=  x \\
	\tradn{(\lambda x M)} &= \lambda x \tradn M \\
	\tradn{(M~N)} &= \tradn M \oc \tradn N \\
	\tradn{(M\es Nx)} &= \tradn M\es{\oc \tradn N}x
\end{align*}
\end{minipage}
\hfill
\begin{minipage}{0.30\textwidth}
\begin{align*}
	\tradv x &= \oc x \\
	\tradv{(\lambda x M)} &= \oc (\lambda x \tradv M) \\
	\tradv{(M~N)} &= \left\{
		\begin{array}{c}
		L\langle P \rangle \tradv N \qquad \text{ if } \tradv M = L \langle \oc P \rangle \\
		\der{\tradv M}~\tradv N \qquad \text{ otherwise }
		\end{array}
	\right. \\
	\tradv{(M\es Nx)} &= \tradv M\es{\tradv N}x
\end{align*}
\end{minipage}
\end{center}

We abusively extend the translations to list contexts: let
$\circ\in\{n,v\}$; then given a list context
$L=\square[M_1/x_1]\dots [M_k/x_k]$, we write
$L^\circ=\square[M_1^\circ/x_1]\dots [M_k^\circ/x_k]$.

We are now ready to define meaningfulness in both systems:
\begin{definition}[\dcbv{} and \dcbn{} meaningfulness~\cite{kag24}]
	\label{def:dcbn-dcbv-meaningfulness}
	We say that a term $M \in \dcbv{}$ is \dcbv-meaningful (resp.
	$M\in\dcbn{}$ and is \dcbn-meaningful) if $T\langle M \rangle
	\tov^* V$ for some value $V$ (resp. $T\langle M \rangle \ton^*
	\lambda xx$) where $T$ is a testing context, defined by: $T :=
	\square \alt T~N \alt (\lambda x T)~N$.
	\footnote{
		Note that as for \dbang, if such a reduction exists, then 
		a weak reduction (analoguous to surface
		reduction) is sufficient, \emph{i.e.} it is not necessary to reduce under
		the $\lambda$s in $\dcbv$ or in the arguments of applications (or explicit
		substitutions) in \dcbn.
	}.
\end{definition}
\begin{example}
\fnl{ Consider the fixpoint $Y_x=(\lambda y x(yy))(\lambda y x(yy))$,
	its translations are:
	\begin{itemize}
		\item $\tradn{(Y_x)}=(\lambda yx\oc(y\oc y))\oc(\lambda y x\oc(y\oc y))$.
		\item $\tradv{(Y_x)}=(\lambda y x(y\oc y))\oc(\lambda y x(y\oc y))$.
	\end{itemize}
	Which are the terms $Y^n_x$ and $Y^v_x$ of
	Example~\ref{ex:terms}.}
\end{example}
It has been shown that \dcbn{} \dcbv{} can be simulated through their encoding
in \dbang{}, and that the meaningfulness of \dcbv{} and \dcbn{} coincides with
the one of \dbang{} (Definition~\ref{def:dbang_meaningful}):

\begin{theorem}~
	\label{thm:trad_meaningfulness}
	\begin{enumerate}
		\item If $M \ton N$ (resp. $M\tov N$) then $\tradn M \todb \tradn N$
		(resp. $\tradv M\todb\tradv N$)~\cite[Lemma 4.6]{bucciarelli2023bang}.
		\item $M$ is \dcbn{}-meaningful iff $M^n$ is meaningful~\cite[Theorem 25]{kag24}.
		\item $M$ is \dcbv{}-meaningful iff $M^v$ is meaningful~\cite[Theorem 30]{kag24}.
	\end{enumerate}
\end{theorem}

We study approximation trees and Taylor expansion for \dcbv{} and \dcbn{} via their translations, in view of meaningfulness.

\subsection{Taylor expansion and  approximation trees for \dcbn\ and \dcbv}

\subsubsection{Taylor expansion for \dcbn{} and \dcbv{}}

\begin{definition}[Resource approximations and Taylor expansion of
	\dcbn]\label{def:dcbn_taylor_approx}~
	
	We define an approximation $\appn$ relation between
		\dbangres and
		\dcbn. Notice that despite the approximants being defined in
		\dbangres, there is no dereliction needed in the case of
		\dcbn. The relation is given by $x\appn x$; if $m\appn M$ then
		$\lambda xm\appn \lambda xM$; and if $m\appn M$ and $n_i\appn
		N$ for any $i\leq k$ then we have $m[n_1,\dots ,n_k]\appn MN$ and
		$m[[n_1,\dots ,n_k]/x]\appn M[N/x]$.

	Taylor expansion is again defined as sets of approximations: $\taylorn
	M=\{m\in\dbangres\mid m\appn M\}$.
\end{definition}
We do not need to define a specific
		resource calculus for \dcbn{} nor \dcbv{}, since \dbangres{}
		semantics precisely subsumes both approximation theories. 

\begin{remark}
	\fnl{The syntax of approximants of \dcbn{} can be given by the
following syntax: $m, n := x \alt m~[n_1, \dots , n_k] \alt \lambda xm
\alt m\es{[n_1, \dots , n_k]}x$}
\end{remark}

\begin{definition}[Resource approximation and Taylor expansion of
	\dcbv]\label{def:dcbv_taylor_approx}
	We define the relation $m\appv M$ for $m\in\dbangres$ and $M\in\dcbv$.
	\begin{itemize}
		\item $[x,\dots ,x]_k\appv x$ for any $k\in\mathbb N$.
		\item $[\lambda xm_1,\dots ,\lambda xm_k]\appv \lambda xM$ if
			$m_i\appv M$ for
			any $i\leq k$.
		\item $ \der mn\appv MN$ if $m\appv M, n\appv N$ and
				$M\notin V$
		\item $mn\appv VN$ if $[m]\appv V$ and $n\appv N$

		\item $m[n/x]\appv M[N/x]$ if $m\appv M$ and $n\appv N$.
	\end{itemize}
	
	Taylor expansion is defined as $\taylorv M=\{m\in\dbangres\mid m\appv
	M\}$. Notice that, so as Taylor expansion commutes with this embedding,
	Taylor approximation also suppresses derelictions redexes, so as the
	expansion also preserves normal forms.
\end{definition}

\begin{example}
	\fnl{
	Approximants can fit both for \dcbn\ and \dcbv\ terms, such as $x[x]\appv xx$ and
	$x[x]\appn xx$. But this is not the case in general, \emph{e.g.} $[x]\appv x$ but
	there is no \dcbn\ term $M$ such that $[x]\appv M$.
}
\end{example}
	
Both in \dcbv\ and \dcbn\, we also define $\taylornf{M}$ as the set
containing the normal forms of resource approximants of $M$. We get the following: 

\begin{lemma}
	\label{lem:taylor_trad_dcbn_dcbv}
	Let $M \in \dcbn$ (resp. $\dcbv$) then $\taylorn M = \taylor{\tradn M}$ (resp. $\tradv M$)
\end{lemma}

In the $\dcbv$ case, notice that the translation of application, as
well as its Taylor expansion, is described by case on its first
component such that the translation (resp. the expansion) of a term in
normal form does not lead to any reducible pattern.
From that we obtain the following properties:

\begin{corollary}\label{cor:tradn_tnf-tradv-tnf}
	Let $M\in\dcbn$ (resp. $\dcbv$) then $\taylornf
	M=\emptyset\leftrightarrow \taylornf{\tradn M}=\emptyset$ (resp. $\tradv M$).
\end{corollary}

\subsubsection{Approximation trees for \dcbn{} and \dcbv{}}

\begin{definition}The syntax of \dcbn\ approximants is given by:
\begin{center}
	$A_n ::=  \bot\alt N_{\lambda} \alt \lambda x A_n
	\qquad\qquad N_{\lambda } := x  \alt N_{\lambda } A_n$
\end{center}
\end{definition}
 
Notice that this syntax contains inductive weak head normal forms,  as it is standard for
Lévy-Longo trees ~\cite{Longo83}. Subterms such as $\lambda\vec x\bot$,
which might correspond to non trivial approximants of non-solvable terms (\emph{e.g.} $\lambda
x\Omega$) are not included in Böhm trees, which correspond to inductive head normal forms.  We could endow the approximants with equations identifying $\lambda\vec x\bot$
to $\bot$, but this is not necessary\footnote{
	In fact, as for \dbang, having $\bt(\lambda xM)=\lambda x\bt(M)$ holding globally is
	convenient for technical reasons.
}
since this trivialization is achieved by Taylor
expansion. 

\begin{definition} The syntax of \dcbv\ approximants is given by:

\begin{minipage}{0.45\textwidth}
$$\begin{array}{lll}
	A_v & ::= & \bot\alt A_{\lambda } \alt \lambda x A_v \alt A_v\es{A_{x\lambda }}x \\
	A_{\lambda} & ::= &  x  \alt A_{\lambda} A_v \alt A_{\lambda
	}\es{A_{x\lambda }}x \\
\end{array}$$		
\end{minipage}
\begin{minipage}{0.54\textwidth}
	$$\begin{array}{lll}
	A_{x\lambda } & ::= & A_{\lambda } A_v   \alt
	 A_{x\lambda }
	\es{A_{x\lambda }}x
\end{array}$$
\end{minipage}

\end{definition}
In \dcbv{}, as well as in the \cbv{} case (and its variant with permutation rules), there is no equivalent to (weak) head normal forms. We define our approximants as the normal forms potentially containing $\bot$. This is similar to the ones used for Böhm trees in \cite{kerinec:bohmTreeCbV} (for a variant of \cbv{} with permutation rules).\\ 

We then define approximation trees in the usual way, setting
$\bt_\circ(M)=\cup\{A_\circ\mid M\to_\circ^*N,A_\circ\sqsubseteq N\}$, for $\circ\in\{n,v\}$,
(and where $M\sqsubseteq N$ means again that $M$ is obtained by replacing subterms of
$N$ by $\bot$). We leave it to the reader to
verify that the proof of Lemma~\ref{lem:ideal} can be adapted to \dcbv\ and
\dcbn, incorporating the distant setting into these standard
results.

In the following, we extend translations $()^n$ and $()^v$ to
approximants by setting $\bot^n=\bot^v=\bot$, and study the
commutation between approximation and the embedding
(Theorem~\ref{thm:bohm-trad}), that will allow us to transport our
Commutation Theorem (Theorem~\ref{thm:taylor-bohm-commute}) into
\dcbv\ and \dcbn. Notice that $M$ is a \dcbn{} (resp. \dcbv{})
approximant if and only if $\tradn M$ (resp. $\tradv M$) is a \dbang{}
approximant.

Now we can prove the weaker form of reverse simulation mentioned in the preamble of
this section. Notice that this result corresponds to the 
embedding notion of Dufour and Mazza~\cite{mazza-dufour}.

\begin{lemma}[Embedding]~\label{lem:trad_embedding}
	\begin{enumerate}
		\item Let $M\in\dcbn$. If $\tradn M\todbf N$, then there is some
			$P\in\dcbn$ such that $M\ton^* P$ and $N\todbf^*\tradn
			P$.
		\item Let $M\in\dcbv$. If $\tradv M\todbf N$, then there is some
			$P\in\dcbv$ such that $M\tov^* P$ and $N\todbf^*\tradv
			P$.
	\end{enumerate}
\end{lemma}

\begin{theorem}~\label{thm:bohm-trad}
	\begin{enumerate}
		\item Let $M\in \dcbn$. $\tradn{(\btn(M))} = \btb(\tradn M)$.
		\item Let $M\in\dcbv$. Then $\tradv{(\btv(M))} = \btb(\tradv M)$.
	\end{enumerate}
\end{theorem}

Thanks to the compatibility of approximation trees and Taylor expansion with the
translations of \dcbn\ and \dcbv\ into \dbang\, we can
apply our commutation result for \dbang\ to both calculi. Although these results
have been well established (to our knowledge, only for non-distant \cbn{}
and \cbv{}), this application illustrates
 that the subsuming power of \dbang\ has been brought in the
approximation theory, which was our purpose. 

\begin{theorem}~
	\label{commutation-thm-translations}
	\begin{enumerate}
		\item Let $M\in\dcbv$. $\taylorv{\bt_v(M)}=\nf{\taylorv M}$.
		\item Let $M\in\dcbn$. $\taylorn{\bt_n(M)}=\nf{\taylorn M}$.
	\end{enumerate}
\end{theorem}

\section{Meaningfulness and Taylor expansion}
\label{sec:meaningfulness}
In \cbn{} and \cbv{}, solvability and scrutability are characterized by non-empty Taylor normal forms~\cite{ER08,CG14}.
We aim to establish an analogous link in \dbang, for which a characterization of
meaningfulness has been provided~\cite{kag24}, but without an approximation theory related
to this notion. The first part of this result is achieved by
Theorem~\ref{thm:tnf-dbang}. The
second part presents signifiant challenges, which  this section aims to
address.

\begin{theorem}\label{thm:tnf-dbang}
	Let $M\in\dbang$. If $M$ is meaningful, then $\taylornf{M}\neq
	\emptyset$.
\end{theorem}

This result is encouraging for our study of Taylor expansion in $\dbang$
framework, as it applies to $\dcbn$ and $\dcbv$ through our previous 
results (Corollary~\ref{cor:tradn_tnf-tradv-tnf}): 
Theorem~\ref{thm:tnf-dbang} applies to both settings.

The converse fails in general: some non-empty Taylor normal forms correspond to
meaningless terms. As mentioned in~\cite{kag24}, some elementary terms, such as
$xx$, are meaningless, whereas $xy$ is not. Their intersection type system can
distinguish between these terms, but it is unlikely that such a distinction can
be captured at the syntactic level using Taylor expansion. 
Another approach would be to restrict ourselves to a (clash-free) fragment of
$\dbang$ excluding patterns that do not make sense from a $\dcbv$ nor a $\dcbn$
discipline (such as $xx$), but again we can exhibit terms such as $(x\oc x)(x\oc
x)$ (
	Recall that $\tradv{(xM)}=\oc x(\tradv M)$ and $\tradn{(Mx)}=\tradn M
	\oc x$.
) that are meaningless but cannot reasonably be assigned an empty Taylor normal
form. 

We prove in the remaining of this section that the result holds independently
for the two sublanguages of \dbang\ consisting of terms translated from
$\tradv{()}$ and $\tradn{()}$.
Our proof employs techniques adapted from \cbn{}~\cite{ER08} and
\cbv~\cite{CG14}, providing an initial characterization of the relationship between meaningfulness and Taylor
expression in a distant setting. Although it is frustrating that we cannot prove
the equivalence once for $\dbang$ and to apply it directly to its fragments;
this limitation 
also raises an open question which we find to be of interest: is there a
significant, bigger fragment of \dbang\ for which the equivalence can be proven
generically? Would this fragment cover terms not coming from a \dcbv\ or \dcbn\
translations? This is, for now, an open question.

We consider two strict subsets of \dbang : $\dbangv$ and $\dbangn$ corresponding
to terms obtained by translating from \dcbn\ and \dcbv, respectively. These
fragments also have the advantage of excluding 
\emph{clashes} - problematic $\dbang$ terms such as $\der{\lambda xM}$ - 
which are often omitted from the analysis~\cite{bucciarelli2023bang,
ehrhard2016bang} (see Remark~\ref{rem:clashes}).

\begin{definition}\

	\dbangn : $M_n := x \mid \lambda x M_n \mid M_n \oc  M_n \mid M_n\es{\oc
	M_n}x$
	
	\dbangv : 
		$M_v := \oc x \mid \oc(\lambda x M_v) 
		\mid L_v\<\lambda x M_v\> M_v \mid L_v\<x\> M_v \mid \der{M_v}M_v \mid M_v\es{M_v}x $

	$L_v:=\square\mid L_v[M_v/x]$

\end{definition}

A simple inspection of the definitions yields the following property:
\begin{lemma}\label{lem:trad_fragments}
	For any $M\in\dcbv, \tradv M\in\dbangv$, and for any $M\in\dcbn, \tradn M\in\dbangn$.
\end{lemma}

Note that the converse holds for $\dbangn$, but not for $\dbangv$. For example
$\der{\oc x}M\in\dbangv$, but no term of $\dcbv$ translates to
this term due of the side condition of $(-)^v$ which ensures the preservation of
normal forms. However, we cannot
exclude these patterns from $\dbangv$ as they can be obtained from some reduction
as shown in Section~\ref{sec:translations}.

We aim to ensure that our fragments are closed under reduction. Otherwise,
a term in $\dbangv$, for example, could reduce in a term in $\dbang$ for which
meaningfulness cannot be guaranteed (such as $xx$ or clashes like
$\der{\lambda xM}$). The following lemma can be proven by a standard induction.
\begin{lemma}
	$\dbangv$ and $\dbangn$ are closed under $\todbf$.
\end{lemma}

\subsection{Meaningfulness and Taylor Expansion for \dcbn}

The case of \dcbn\ is relatively easy to handle, as we can adapt to the distant
case the following properties; which correspond to well-known features of $\lambda$
calculus: resource normal form correspond to head normal forms and terms with head normal
forms are meaningful.

\begin{lemma}
	\label{lemma:dbangn-nf-shape}
	The normal forms of \dbangn{} are of shape $\lambda x_1\dots x_k (x) \oc N_1\dots \oc
	N_l$, where $k,l\in\mathbb N$.
\end{lemma}

Naturally, full normal form requires the $N_i$ to be in normal form too, but as
we shall see, this is not relevant for studying meaningfulness, as these terms
will be erased by an appropriate testing context. 
Previous observations can be brought at a resource level.
\begin{lemma}\label{lem:dbangn_nf}
	Let $m\appn M$ with $M\in\dbangn$. If $\nf m\neq\emptyset$, it is of shape $\lambda
	x_1\dots x_k(x)\bar n_1\dots \bar n_l$.
\end{lemma}

We are now able to state the theorem establishing the classical link between Taylor
expansion and meaningfulness in the case of $\dbangn$.
\begin{theorem}\label{thm:dbangn_tnf}
	Let $M\in\dbangn$. If $\taylornf M\neq\emptyset$, then $M$ is
	meaningful.
\end{theorem}

\begin{corollary}\label{cor:dcbn_meaningful}
	For any $M$ in \dcbn, $\taylornf M\neq\emptyset$ if and only if $M$ is
	meaningful.
\end{corollary}

\subsection{Meaningfulness and Taylor expansion for \dcbv}
The \dcbv{} case requires a finer analysis of normal forms. 
We first characterize normal forms of $\dbangv$ and their counterparts in resource
calculus. The following definition and lemma are derived by a standard induction over the
syntax of $\dbangv$.

\begin{definition}\label{def:syntax_nf_dbangv}
	The normal forms of \dbangv are described by the following syntax:

	$B:=B_\oc\mid\oc\lambda xB\mid \oc x\mid L\<B\> \qquad\qquad\qquad L:= \square\mid L[B_\oc/x]$

	$B_\oc :=L\<x\>B\mid \der{B_\oc}B\mid L\<B_\oc\>$

\end{definition}

\begin{lemma}\label{lem:syntax_nf_dbangvres}
	Let $m\app M\in\dbangv$. If $\nf m\neq \emptyset$ then it is of the
	following form:

	$b:= b_\oc \mid [\lambda xb,\dots ,\lambda xb]\mid [x,\dots ,x]\mid l\<b\> \quad~ l:= \square\mid l[b_\oc/x]$

	$b_\oc :=l\<x\>b\mid \der{b_\oc}b\mid l\<b_\oc\>$

\end{lemma}

We consider a family of terms which are
suitable for providing an appropriate testing context for any term of $\dbangv$
with non-empty Taylor normal form, in which it eventually reduces to a value. We define
$\circ_0=\lambda x_0x_0$ and $\circ_{k+1}=\lambda x_{k+1}\oc\circ_k$.

We establish the testing context with the following lemma. The proof 
\iflong (see Appendix~\ref{appendix:meaningfulness}) 
\else (\cite{chardonnet2026approximationtheorydistantbang})
\fi
follows Carraro and Guerrieri's
technique of mutual induction~\cite[Lemmas 26 and 27]{CG14}.

\begin{lemma}\label{lem:mutual1}
	Let $\{x_1,\dots ,x_n\}$ be a set of variables and $M\in B$ 
	(Definition~\ref{def:syntax_nf_dbangv})
	with
	$\fv(M)\subseteq\{x_1,\dots ,x_n\}$. There exists $c\in\mathbb N$ such
	that for any $k_1,\dots ,k_n\geq c$ we have
	$M\{\circ_{k_1}/x_1\dots \circ_{k_n}/x_n\}\todbf^* \oc P$ for some $P$.
\end{lemma}
We can now state the central theorem of this section. 
\begin{theorem}
	\label{thm:dbangv_tnf}
	Let $M\in\dbangv$. If $\taylornf M\neq \emptyset$, then $M$ is
	meaningful.
\end{theorem}

\begin{corollary}
	\label{cor:dcbv-meaningful-iff-taylor}
	Let $M\in\dcbv$ then $M$ is meaningful if and only if $\taylornf
	M\neq\emptyset$.
\end{corollary}

\section{Conclusion and Discussions}
\label{sec:conclusion}
We developed an approximation theory for the distant \dbang-calculus, defining approximation trees and Taylor expansion and relating them to meaningfulness  and to each other. These results are part of a wider effort to generalize the
theory of the \cbn{} and \cbv{} $\lambda$-calculi. 
And indeed, we retrieve similar results in distant \cbn{} and distant \cbv{}, via translations in \dbang.

Future work includes extending meaningfulness to proof structures, pursuing Dufour and
Mazza's work~\cite{mazza-dufour} using non-inductive variants of classical approximation
notions (in particular, their Böhm trees do not have an actual tree structure since
approximations cannot be described through an tree-like inductive syntax). Also, we mentioned in
Section~\ref{sec:meaningfulness} an open question about the possibility to
characterize a significant fragment of \dbang\ for which meaningful terms
coincide with terms having a non-empty Taylor normal form; this line of work
should be explored to develop the general understanding of \dbang.

\bibliography{bibli.bib}

@inproceedings{ehrhard2016bang,
author = {Ehrhard, Thomas and Guerrieri, Giulio},
title = {The Bang Calculus: an untyped lambda-calculus generalizing call-by-name and call-by-value},
year = {2016},
isbn = {9781450341486},
publisher = {Association for Computing Machinery},
address = {New York, NY, USA},
url = {https://doi.org/10.1145/2967973.2968608},
doi = {10.1145/2967973.2968608},
booktitle = {Proceedings of the 18th International Symposium on Principles and Practice of Declarative Programming},
pages = {174–187},
numpages = {14},
location = {Edinburgh, United Kingdom},
series = {PPDP '16}
}

@Book{Amadio98,
  author    = {Roberto Amadio and Pierre-Louis Curien},
  title     = {Domains and Lambda Calculi},
  publisher = {CUP},
  series    = {CTTCS},
  year= {1998}
}

@article{Longo83,
title = {Set-theoretical models of $\lambda$-calculus: theories, expansions, isomorphisms},
journal = {Annals of Pure and Applied Logic},
volume = {24},
number = {2},
pages = {153-188},
year = {1983},
issn = {0168-0072},
doi = {https://doi.org/10.1016/0168-0072(83)90030-1},
url = {https://www.sciencedirect.com/science/article/pii/0168007283900301},
author = {Giuseppe Longo}
}

@book{barendregt,
  booktitle        = {The Lambda Calculus: Its Syntax and Semantics},
  author       = {Barendregt, Henk},
  year         = {1984},
  publisher    = {North-Holland},
  series       = {Studies in Logic and the Foundations of Mathematics},
  volume       = {103},
  isbn         = {978-0-444-87508-2},
  address      = {Amsterdam}
}

@incollection{barendregt-bohm,
title = {The Type Free Lambda Calculus.},
editor = {Jon Barwise},
series = {Studies in Logic and the Foundations of Mathematics},
publisher = {Elsevier},
volume = {90},
pages = {1091-1132},
year = {1977},
booktitle = {HANDBOOK OF MATHEMATICAL LOGIC},
issn = {0049-237X},
doi = {https://doi.org/10.1016/S0049-237X(08)71129-7},
url = {https://www.sciencedirect.com/science/article/pii/S0049237X08711297},
author = {Henk P. Barendregt}
}

@article{ER08,
  author       = {Thomas Ehrhard and
                  Laurent Regnier},
  title        = {Uniformity and the Taylor expansion of ordinary lambda-terms},
  journal      = {Theor. Comput. Sci.},
  volume       = {403},
  number       = {2-3},
  pages        = {347--372},
  year         = {2008},
  url          = {https://doi.org/10.1016/j.tcs.2008.06.001},
  doi          = {10.1016/J.TCS.2008.06.001},
  bibsource    = {dblp computer science bibliography, https://dblp.org}
}

@inproceedings{CG14,
  author       = {Alberto Carraro and
                  Giulio Guerrieri},
  editor       = {Anca Muscholl},
  title        = {A Semantical and Operational Account of Call-by-Value Solvability},
  booktitle    = {Foundations of Software Science and Computation Structures - 17th
                  International Conference, {FOSSACS} 2014, Held as Part of the European
                  Joint Conferences on Theory and Practice of Software, {ETAPS} 2014,
                  Grenoble, France, April 5-13, 2014, Proceedings},
  series       = {Lecture Notes in Computer Science},
  volume       = {8412},
  pages        = {103--118},
  publisher    = {Springer},
  year         = {2014},
  url          = {https://doi.org/10.1007/978-3-642-54830-7\_7},
  doi          = {10.1007/978-3-642-54830-7\_7},
  bibsource    = {dblp computer science bibliography, https://dblp.org}
}

@article{bucciarelli2023bang,
title = {The bang calculus revisited},
journal = {Information and Computation},
volume = {293},
pages = {105047},
year = {2023},
issn = {0890-5401},
doi = {https://doi.org/10.1016/j.ic.2023.105047},
url = {https://www.sciencedirect.com/science/article/pii/S0890540123000500},
author = {Antonio Bucciarelli and Delia Kesner and Alejandro Ríos and Andrés Viso}
}

@inproceedings{agk24d,
  author       = {Victor Arrial and
                  Giulio Guerrieri and
                  Delia Kesner},
  editor       = {Christoph Benzm{\"{u}}ller and
                  Marijn J. H. Heule and
                  Renate A. Schmidt},
  title        = {The Benefits of Diligence},
  booktitle    = {Automated Reasoning - 12th International Joint Conference, {IJCAR}
                  2024, Nancy, France, July 3-6, 2024, Proceedings, Part {II}},
  series       = {Lecture Notes in Computer Science},
  volume       = {14740},
  pages        = {338--359},
  publisher    = {Springer},
  year         = {2024},
  url          = {https://doi.org/10.1007/978-3-031-63501-4\_18},
  doi          = {10.1007/978-3-031-63501-4\_18},
  bibsource    = {dblp computer science bibliography, https://dblp.org}
}

@InProceedings{kag24,
  author =	{Kesner, Delia and Arrial, Victor and Guerrieri, Giulio},
  title =	{{Meaningfulness and Genericity in a Subsuming Framework}},
  booktitle =	{9th International Conference on Formal Structures for Computation and Deduction (FSCD 2024)},
  pages =	{1:1--1:24},
  series =	{Leibniz International Proceedings in Informatics (LIPIcs)},
  ISBN =	{978-3-95977-323-2},
  ISSN =	{1868-8969},
  year =	{2024},
  volume =	{299},
  editor =	{Rehof, Jakob},
  publisher =	{Schloss Dagstuhl -- Leibniz-Zentrum f{\"u}r Informatik},
  address =	{Dagstuhl, Germany},
  URL =		{https://drops.dagstuhl.de/entities/document/10.4230/LIPIcs.FSCD.2024.1},
  URN =		{urn:nbn:de:0030-drops-203305},
  doi =		{10.4230/LIPIcs.FSCD.2024.1}
}

@InProceedings{structural-lambda-calculus,
author="Accattoli, Beniamino
and Kesner, Delia",
editor="Dawar, Anuj
and Veith, Helmut",
title="The Structural $\lambda$-Calculus",
booktitle="Computer Science Logic",
year="2010",
publisher="Springer Berlin Heidelberg",
address="Berlin, Heidelberg",
pages="381--395",
isbn="978-3-642-15205-4"
}

@InProceedings{accattoli:LIPIcs.RTA.2012.6,
  author =	{Accattoli, Beniamino},
  title =	{{An Abstract Factorization Theorem for Explicit Substitutions}},
  booktitle =	{23rd International Conference on Rewriting Techniques and Applications (RTA'12)},
  pages =	{6--21},
  series =	{Leibniz International Proceedings in Informatics (LIPIcs)},
  ISBN =	{978-3-939897-38-5},
  ISSN =	{1868-8969},
  year =	{2012},
  volume =	{15},
  editor =	{Tiwari, Ashish},
  publisher =	{Schloss Dagstuhl -- Leibniz-Zentrum f{\"u}r Informatik},
  address =	{Dagstuhl, Germany},
  URL =		{https://drops.dagstuhl.de/entities/document/10.4230/LIPIcs.RTA.2012.6},
  URN =		{urn:nbn:de:0030-drops-34813},
  doi =		{10.4230/LIPIcs.RTA.2012.6}
}

@article{Vaux19,
  author       = {Lionel Vaux},
  title        = {Normalizing the Taylor expansion of non-deterministic \{{\textbackslash}lambda\}-terms,
                  via parallel reduction of resource vectors},
  journal      = {Log. Methods Comput. Sci.},
  volume       = {15},
  number       = {3},
  year         = {2019},
  url          = {https://doi.org/10.23638/LMCS-15(3:9)2019},
  doi          = {10.23638/LMCS-15(3:9)2019},
  bibsource    = {dblp computer science bibliography, https://dblp.org}
}

@InProceedings{mazza-dufour,
  author =	{Dufour, Alo\"{y}s and Mazza, Damiano},
  title =	{{B\"{o}hm and Taylor for All!}},
  booktitle =	{9th International Conference on Formal Structures for Computation and Deduction (FSCD 2024)},
  pages =	{29:1--29:20},
  series =	{Leibniz International Proceedings in Informatics (LIPIcs)},
  ISBN =	{978-3-95977-323-2},
  ISSN =	{1868-8969},
  year =	{2024},
  volume =	{299},
  editor =	{Rehof, Jakob},
  publisher =	{Schloss Dagstuhl -- Leibniz-Zentrum f{\"u}r Informatik},
  address =	{Dagstuhl, Germany},
  URL =		{https://drops.dagstuhl.de/entities/document/10.4230/LIPIcs.FSCD.2024.29},
  URN =		{urn:nbn:de:0030-drops-203582},
  doi =		{10.4230/LIPIcs.FSCD.2024.29}
}

@phdthesis{kerinec:phd,
  TITLE = {{A story of lambda-calculus and approximation}},
  AUTHOR = {Kerinec, Axel},
  URL = {https://theses.hal.science/tel-04624826},
  NUMBER = {2023PA131075},
  SCHOOL = {{Universit{\'e} Paris-Nord - Paris XIII}},
  YEAR = {2023},
  MONTH = Jun,
  TYPE = {Theses},
  HAL_ID = {tel-04624826},
  HAL_VERSION = {v1},
}

@article{kerinec:bohmTreeCbV,
    title      = {Revisiting Call-by-value B\"ohm trees in light of their Taylor expansion},
    author     = {Axel Kerinec and Giulio Manzonetto and Michele Pagani},
    url        = {https://lmcs.episciences.org/4817},
    doi        = {10.23638/LMCS-16(3:6)2020},
    journal    = {Logical Methods in Computer Science},
    issn       = {1860-5974},
    volume     = {Volume 16, Issue 3},
    eid        = 6,
    year       = {2020},
    month      = {Jul},
}

@inproceedings{Barendregt84,
author = {Barendregt, Henk (Hendrik) and Eekelen, Marko and Glauert, John and Kennaway, Richard and Plasmeijer, Marinus and Sleep, Michael},
year = {1987},
month = {01},
pages = {141-158},
title = {Term Graph Rewriting.}
}

@article{Ariola70,
author = {Ariola, Zena},
year = {1970},
month = {02},
pages = {},
title = {Relating Graph and Term Rewriting via Böhm Models},
volume = {7},
isbn = {978-3-540-56868-1},
journal = {Applicable Algebra in Engineering, Communication and Computing},
doi = {10.1007/3-540-56868-9_15}
}

@article{Ariola02,
title = {Skew confluence and the lambda calculus with letrec},
journal = {Annals of Pure and Applied Logic},
volume = {117},
number = {1},
pages = {95-168},
year = {2002},
issn = {0168-0072},
doi = {https://doi.org/10.1016/S0168-0072(01)00104-X},
url = {https://www.sciencedirect.com/science/article/pii/S016800720100104X},
author = {Zena M. Ariola and Stefan Blom}
}

@InProceedings{ehrhard:Collapsing,
  author =	{Ehrhard, Thomas},
  title =	{{Collapsing non-idempotent intersection types}},
  booktitle =	{Computer Science Logic (CSL'12) - 26th International Workshop/21st Annual Conference of the EACSL},
  pages =	{259--273},
  series =	{Leibniz International Proceedings in Informatics (LIPIcs)},
  ISBN =	{978-3-939897-42-2},
  ISSN =	{1868-8969},
  year =	{2012},
  volume =	{16},
  editor =	{C\'{e}gielski, Patrick and Durand, Arnaud},
  publisher =	{Schloss Dagstuhl -- Leibniz-Zentrum f{\"u}r Informatik},
  address =	{Dagstuhl, Germany},
  URL =		{https://drops.dagstuhl.de/entities/document/10.4230/LIPIcs.CSL.2012.259},
  URN =		{urn:nbn:de:0030-drops-36776},
  doi =		{10.4230/LIPIcs.CSL.2012.259}
}

@techreport{Boudol83,
  TITLE = {{Computational semantics of terms rewriting systems}},
  AUTHOR = {Boudol, G{\'e}rard},
  URL = {https://inria.hal.science/inria-00076366},
  TYPE = {Research Report},
  NUMBER = {RR-0192},
  PAGES = {98},
  INSTITUTION = {{INRIA}},
  YEAR = {1983},
  HAL_ID = {inria-00076366},
  HAL_VERSION = {v1},
}

@inproceedings{Carraro14,
  author    = {Albero Carraro and
   Giulio Guerrieri},
  editor    = {Anca Muscholl},
  title     = {A Semantical and Operational Account of Call-by-Value Solvability},
  booktitle = {Foundations of Software Science and Computation Structures - 17th
   International Conference, {FOSSACS} 2014, Held as Part of the European
   Joint Conferences on Theory and Practice of Software, {ETAPS} 2014,
   Grenoble, France, April 5-13, 2014, Proceedings},
  series    = {Lecture Notes in Computer Science},
  volume    = {8412},
  pages     = {103--118},
  publisher = {Springer},
  year= {2014},
  url = {https://doi.org/10.1007/978-3-642-54830-7\_7},
  doi = {10.1007/978-3-642-54830-7\_7},
  bibsource = {dblp computer science bibliography, https://dblp.org}
}

@InProceedings{alberto-Giulio-CbV-Solvability,
author="Carraro, Alberto
and Guerrieri, Giulio",
editor="Muscholl, Anca",
title="A Semantical and Operational Account of Call-by-Value Solvability",
booktitle="Foundations of Software Science and Computation Structures",
year="2014",
publisher="Springer Berlin Heidelberg",
address="Berlin, Heidelberg",
pages="103--118",
isbn="978-3-642-54830-7"
}

@InProceedings{levy-cbpv,
author="Levy, Paul Blain",
editor="Girard, Jean-Yves",
title="Call-by-Push-Value: A Subsuming Paradigm",
booktitle="Typed Lambda Calculi and Applications",
year="1999",
publisher="Springer Berlin Heidelberg",
address="Berlin, Heidelberg",
pages="228--243",
isbn="978-3-540-48959-7"
}

@InProceedings{ehrhard-cbpv-ll,
author="Ehrhard, Thomas",
editor="Thiemann, Peter",
title="Call-By-Push-Value from a Linear Logic Point of View",
booktitle="Programming Languages and Systems",
year="2016",
publisher="Springer Berlin Heidelberg",
address="Berlin, Heidelberg",
pages="202--228",
isbn="978-3-662-49498-1"
}

@article{ehrhard-uniformity-taylor,
title = {Uniformity and the Taylor expansion of ordinary lambda-terms},
journal = {Theoretical Computer Science},
volume = {403},
number = {2},
pages = {347-372},
year = {2008},
issn = {0304-3975},
doi = {https://doi.org/10.1016/j.tcs.2008.06.001},
url = {https://www.sciencedirect.com/science/article/pii/S0304397508004064},
author = {Thomas Ehrhard and Laurent Regnier}
}

@article{arrial-quantitative,
author = {Arrial, Victor and Guerrieri, Giulio and Kesner, Delia},
title = {Quantitative Inhabitation for Different Lambda Calculi in a Unifying Framework},
year = {2023},
issue_date = {January 2023},
publisher = {Association for Computing Machinery},
address = {New York, NY, USA},
volume = {7},
number = {POPL},
url = {https://doi.org/10.1145/3571244},
doi = {10.1145/3571244},
journal = {Proc. ACM Program. Lang.},
month = jan,
articleno = {51},
numpages = {31}
}

@InProceedings{arrial-diligence,
author="Arrial, Victor
and Guerrieri, Giulio
and Kesner, Delia",
editor="Benzm{\"u}ller, Christoph
and Heule, Marijn J.H.
and Schmidt, Renate A.",
title="The Benefits of Diligence",
booktitle="Automated Reasoning",
year="2024",
publisher="Springer Nature Switzerland",
address="Cham",
pages="338--359",
isbn="978-3-031-63501-4"
}

@inproceedings{arrial-genericity,
author = {Arrial, Victor and Guerrieri, Giulio and Kesner, Delia},
title = {Genericity Through Stratification},
year = {2024},
isbn = {9798400706608},
publisher = {Association for Computing Machinery},
address = {New York, NY, USA},
url = {https://doi.org/10.1145/3661814.3662113},
doi = {10.1145/3661814.3662113},
booktitle = {Proceedings of the 39th Annual ACM/IEEE Symposium on Logic in Computer Science},
articleno = {5},
numpages = {15},
location = {Tallinn, Estonia},
series = {LICS '24}
}

@inproceedings{laird-weighted-rel,
author = {Laird, Jim and Manzonetto, Giulio and McCusker, Guy and Pagani, Michele},
title = {Weighted Relational Models of Typed Lambda-Calculi},
year = {2013},
isbn = {9780769550206},
publisher = {IEEE Computer Society},
address = {USA},
url = {https://doi.org/10.1109/LICS.2013.36},
doi = {10.1109/LICS.2013.36},
booktitle = {Proceedings of the 2013 28th Annual ACM/IEEE Symposium on Logic in Computer Science},
pages = {301–310},
numpages = {10},
series = {LICS '13}
}

@article{ehrhard-differential,
title = {The differential lambda-calculus},
journal = {Theoretical Computer Science},
volume = {309},
number = {1},
pages = {1-41},
year = {2003},
issn = {0304-3975},
doi = {https://doi.org/10.1016/S0304-3975(03)00392-X},
url = {https://www.sciencedirect.com/science/article/pii/S030439750300392X},
author = {Thomas Ehrhard and Laurent Regnier}
}

@book{Girard_Lafont_Regnier_1995,
editor={Girard J-Y, Lafont Y, Regnier L},
 place={Cambridge}, 
 series={London Mathematical Society Lecture Note Series}, 
 title={Advances in Linear Logic}, 
 publisher={Cambridge University Press}, 
 year={1995}, 
 collection={London Mathematical Society Lecture Note Series}
}

@article{Accattoli_2012,
   title={Preservation of Strong Normalisation modulo permutations for the structural lambda-calculus},
   volume={Volume 8, Issue 1},
   ISSN={1860-5974},
   url={http://dx.doi.org/10.2168/LMCS-8(1:28)2012},
   DOI={10.2168/lmcs-8(1:28)2012},
   journal={Logical Methods in Computer Science},
   publisher={Centre pour la Communication Scientifique Directe (CCSD)},
   author={Accattoli, Beniamino and Kesner, Delia},
   year={2012},
   month=mar }

@article{girard-ll,
title = {Linear logic},
journal = {Theoretical Computer Science},
volume = {50},
number = {1},
pages = {1-101},
year = {1987},
issn = {0304-3975},
doi = {https://doi.org/10.1016/0304-3975(87)90045-4},
url = {https://www.sciencedirect.com/science/article/pii/0304397587900454},
author = {Jean-Yves Girard}
}

@article{chardonnet2026approximationtheorydistantbang,
      title={Approximation theory for distant Bang calculus}, 
      author={Kostia Chardonnet and Jules Chouquet and Axel Kerinec},
      year={2026},
      eprint={2601.05199},
      archivePrefix={arXiv},
      primaryClass={cs.LO},
      url={https://arxiv.org/abs/2601.05199},
      note={Long version with technical appendix},
      journal={}
}

\iflong 
 \newpage
 \appendix
  \section{Proofs of Section~\ref{sec:dbang}}
\label{appendix:dbang}

\begin{lemma}\label{lem:simu_dbr_one}
	If $M\todb N$, then for any $m\app M$, either $m\todbr \emptyset$, or there is some $n\app N$ such
	that $m\todbr n$.
\end{lemma}

\begin{proof}[Proof of Lemma~\ref{lem:simu_dbr_one}]
	By induction on the definition $\todb$:
	\begin{itemize}
		\item If $m\app\der{L\<\oc N\>}$, then there exist $k\in\mathbb
			N, n_1,\dots ,n_k\app N$, and $l\app L$ such that
			$m=\der{l\<[n_1,\dots ,n_k]\>}$, and then $m\todb \emptyset$ if $k\neq
			1$, otherwise $m\todb l\<n_1\>\app L\<N\>$.
		\item If $m\app(L\<\lambda xN\>P)$, then $m=l\<\lambda xn\>p$
			for some $l\app L, n\app N$ and $p\app P$. Then,
			$m\todbr l\<n[p/x]\>\app L\<N[P/x]\>$.
		\item If $m\app N[L\<\oc P\>/x]$ then there are $k\in\mathbb N,
			p_1,\dots ,p_k\app P, n\app N,l\app N$ such that
			$m=n[l\<[p_1,\dots ,p_k]/x]$. Then $m\todbr \emptyset$ if $d_x(m)\neq
			k$, and otherwise for any $\sigma\in P_k$, $m\todbr
			l\<n\{p_{\sigma(1)}/x_1,\dots ,p_{\sigma(k)}/x_k\}\>$. \qedhere
	\end{itemize}
\end{proof}

This simulation property can be extended to surface contexts:
\begin{lemma}
	\label{lem:simu_dbr}
	If $M\todbs N$, then for any $m\app M$, either $m\todbrs \emptyset$ or there is some $n\app N$ such
	that $m\todbrs n$.
\end{lemma}
\begin{proof}[Proof of Lemma~\ref{lem:simu_dbr}]
	Let $M=S\<M'\>$ and $N=S\<N'\>$ with $M\todb N$; and then $m=s\<m'\>$
	for $s\app S$.
	By induction on surface contexts:
	\begin{itemize}
		\item $S=\square$. Then $M\todb N$, we apply
			Lemma~\ref{lem:simu_dbr_one}.
		\item $S=\lambda x S'$. Then $M=\lambda xS'\<M'\>$ and
			$N=\lambda x S'\<N'\>$. By induction hypothesis, either
			$m'\todbrs \emptyset$, and then	$\lambda
			xs'\<m'\>\todbrs \emptyset$, either
			we have some $n'\app N'$ such that $m'\todbrs n'$, and
			then $\lambda xs\<m'\>\todbr \lambda xs\<n'\>$ by
			definition of resource surface reduction.
	\end{itemize}
	All remainder cases are similar, since for any resource surface context
	$s$ and term $n\app N$, there exist $S$ such that $s\<n\>\app
	S\<N\>$ which enables the induction hypothesis. Again, this fails for
	full contexts. 
\end{proof}

We can obtain a symmetric property with the same arguments. We could qualify,
following Dufour and Mazza~\cite{mazza-dufour}, the previous lemma as
\emph{push forward}, and the next lemmas as \emph{pull back}.
\begin{lemma}
	\label{lem:simu_dbr_bis}
	If $M\todbs N$, then for any $n\app N$, there is some $m\app M$ such
	that $m\todbrs n$.
\end{lemma}

\begin{proof}[Proof of Lemma~\ref{lem:simu_dbr_bis}]
	Let $M = S\langle M'\rangle, N = S\langle N'\rangle$ with $M' \todb N'$ then
	we can reason by induction on $S$.
	\begin{itemize}
		\item $S = \square$. Then $M \todb N$ and by a straightforward case
		analysis on $\todb$.
		\item $S = \lambda xS'$. Then $M = \lambda x S'\langle M'\rangle, N =
		\lambda x S'\langle N'\rangle$ and $m = \lambda x s' \langle m' \rangle$
		with $s'\app S'$. By induction hypothesis there exists $m' \app M'$ such
		that $m' \todbr n'$ then let $m = \lambda x s'\langle m'\rangle$ we
		have that $m \todbrs n$ by closure.
	\end{itemize}
	All the remainder cases are similar.
\end{proof}

\begin{proof}[Proof of Lemma~\ref{lem:simu_dbr_full}]
	We proceed by induction on the exponential depth where the reduction
	occurs.	By lemma~\ref{lem:simu_dbr}, and since $\todbrs\subseteq\totodbr$, the
	only case we need to show is the full context closure, where $M=F\<M'\>=\oc
	F'\<M'\>$ and $N=F\<N'\>$. 

	Then, $m=[p_1,\dots ,p_k]$ for some $k\in\mathbb N$, where $p_i\app
	F'\<M'\>$.
	
	By induction hypothesis, either some $p_i$ reduces to $\emptyset$, and then also
	$m\totodbr \emptyset$, either for any $i\leq k$, there is some $p_i'\app
	F'\<N'\>$ such that $p_i\totodbr p_i'$. Then,
	$m\totodbr[p_1',\dots ,p_k']\app \oc F'\<N'\>=N$.
\end{proof}

\begin{proof}[Proof of Lemma~\ref{lem:bohm-taylor:equiv-imply-same-nfTaylor}]
	By Theorem~\ref{thm:simu_dbang_full}, we immediately have $\taylornf{M}\subseteq\taylornf{N}$.
	Then, let $n\in\taylornf N$. By iteration of Lemma~\ref{lem:simu_dbr_full_bis}, there
	is some $m\in\taylor M$ such that $m\totodbr^*n$. Then $n$ must
	be in $\taylornf M$. 
\end{proof}

\begin{proof}[Proof of Lemma~\ref{lem:taylor:approx-in-reduces-bis}]
	Consider $m_0\app M$ such that $m_0\todbrf^* m$. We reason by induction
	over the exponential depth under which the reduction occurs. If this
	depth is $0$, then $m_0\todbrs^*m$ and we apply (iteratively)
	Lemma~\ref{lem:simu_dbr_n} to conclude.

	Otherwise, we use the factorization property : by
	Proposition~\ref{prop:standardisation_resources}, we have some $n$ such
	that $m_0\todbrs^*n= s\<\bar n_1,\dots ,\bar n_k\>\totodbr^*s\<\bar n_1',\dots ,\bar n_k'\>=m$. By
	Lemma~\ref{lem:simu_dbr_n}, we have some $N=S\<\oc N_1,\dots ,\oc N_k\>$ such that $M\todbs^*N$ and
	$n\app N$ (thus $n_i\app N_i$). Then, for any $i\leq k$, $\bar
	n_i=[n_{i,1},\dots ,n_{i,l_i}]$ and $\bar n_i'=[n_{i,1}',\dots ,n_{i,l_i}']$ (in
	normal form, since these are subterms of $m$) with
	$n_{i,j}\app N_i$ and $n_{i,j}\todbrf^*n_{i,j}'$. This reduction occurs
	under an exponential (\emph{i.e.} inside a multiset), we then can apply
	our induction hypothesis to assert that there is some $N_i'$ such that
	$N_i\todbf N_i'$ and $n_{i,j}'\app N_i'$. We can conclude, by setting
	$M'=S\<\oc N_1',\dots ,\oc N_k'\>$.
\end{proof}

\subsection{Proof of Section~\ref{subsec:bohm-dbang}}

\begin{proof}[Proof of Lemma~\ref{lem:bang:bohm:reduce-have-at-least-same-approxs}]
	We show the lemma holds for one-step reductions, assuming $M\todbf N$. The closure is easily
	obtained  by induction on the number of steps. By definition,
	$M=A\{P_1/\bot_1,\dots ,P_k/\bot_k\}$, and by Lemma~\ref{lem:approx_bot}, there
	must be some $j\in\{1,\dots ,k\}$ such that $N=A(\{P_i/\bot_i\})_{i\neq
	j}\{P'_j/\bot_j\}$ with $P_j\todbf P_j'$. Consequently, we have
	$A\smaller N$.
\end{proof}

\begin{proof}[Proof of Lemma~\ref{lem:bang:bohm:reduce-have-same-approxs}]
	From Definition \ref{def:bang:bohm:approx-set} we deduce
	$\setApprox(N) \subseteq \setApprox(M)$, and from Lemma \ref{lem:bang:bohm:reduce-have-at-least-same-approxs} the other inclusion. 
\end{proof}

\begin{proof}[Proof of Lemma~\ref{lem:ideal}]
	The downwards closure is by definition of $\setApprox(M)$. For
	directedness, let us assume $A_1, A_2 \in
			\setApprox(M)$, we show that there exists $A_3 \in \setApprox(M)$ such
			that $A_1 \smaller A_3$ and $A_2 \smaller A_3$, by
			induction on $A_1$.

	\begin{itemize}
		\item If $A_1 = x$ then $M\todbf^*x$ and
			$A_2 = x$ or $\bot$ (Remark~\ref{rem:approx-normal}), then we set $A_3=x$.
		\item If $A_1 = \bot$, by definition of $\smaller$ we
			have $A_1 \smaller A_2$, we set $A_3=A_2$.
		\item If $A_1 = \lambda x A_1'$, by
			Lemma~\ref{lem:approx_bot} we have
			$M \todbf^* \lambda x N$ and $A_1' \in
			\setApprox(N)$.
			Then, by
			Lemma~\ref{lem:bang:bohm:reduce-have-same-approxs}
			we have $A_2 \in \setApprox(\lambda x N)$ so
			either $A_2 = \bot$ 
			and we set $A_3=A_2$ or $\lambda x A_2'$ so
			$A_2' \in \setApprox(N)$.
			Then, by induction hypothesis, there
			is some $A_3'$ such that $A_1'\smaller A_3'$
			and $A_2'\smaller A_3'$. We then set
			$A_3=\lambda xA_3'$.
		\item If $A_1 = \oc A_1'$ we reason as in the previous case.
		\item If $A_1 = \der{A'_1}$.
			Then $M \todbf^* \der{N}$ and $A' \in \setApprox(N)$.
			By Lemma~\ref{lem:bang:bohm:reduce-have-same-approxs} we
			have $A_2 \in \setApprox(\der{N})$ so either
			$A_2 = \bot$ or $A_2 = \der{A_2'}$ such that
			$A_2' \in \der{A_\oc}$, and therefore we can use
			our induction hypothesis to obtain an
			upper bound $A_3'$ of $A_1'$ and $A_2'$ and set
			$A_3=\der{A_3'}$, which is indeed an upper bound
			of $\{A_1,A_2\}$ by contextual closure of the
			approximation order.
		\item If $A_1 = A'_1\es{A''_1}{x}$,
			then $M\todbf^* N'\es{N''}{x}= N$
			with $A_1' \in \setApprox(N')$ and $A_1'' \in
			\setApprox(N'')$.
			$A''_1$ is not of shape $L\<\oc-\>$, then
			neither is $N''$. Again,
			$A_2\in\setApprox(N_1)$, then
			$A_2=A_2'[A_2''/x]$ with $A_2'\smaller N'$
			and $A_2''\smaller N''$. By induction
			hypothesis, we can find $A_3'$ an upper bound of
			$\{A_1',A_2'\}$ and $A_3''$ an upper bound of
			$\{A_1'',A_2''\}$. We conclude by setting
			$A_3=A_3'[A_3''/x]$.
		\item If $A_1=A_1'A_1''$, we reason similarly, but using the
			fact that $A_1'$ cannot be of shape $L\<\lambda x -\>$. \qedhere
	\end{itemize}
\end{proof}


\subsection{Proof of Section~\ref{subsec:bang:commutation}}

\begin{proof}[Proof of Lemma~\ref{lem:approx-taylor-subset}]
	By induction on $A$. If $A=\bot$ then $\taylor
	A=\emptyset\subseteq\taylor M$. If $A=x$ then $M=x$ and $\taylor
	A=\taylor M=\{x\}$. If $A=\lambda xA'$, then $M=\lambda xM'$ with $A'\smaller
	M'$. By induction hypothesis, $\taylor{ A'}\subseteq\taylor {M'}$. Then,
	$\taylor A=\{\lambda xa'\mid a'\in\taylor {A'}\}\subseteq\taylor M=\{\lambda
	xm'\mid m'\in\taylor {M'}\}$. The other cases are treated similarly by routine
	induction.
\end{proof}

\begin{proof}[Proof of Lemma~\ref{lem:approx_subset_nft}]
	We have some $M'$ such that $M\todbf^* M'$ and $A\smaller M'$. We have
	that $\nf{\taylor M}=\nf{\taylor{M'}}$ by Lemma~\ref{lem:bohm-taylor:equiv-imply-same-nfTaylor}.
	By Lemma~\ref{lem:approx-taylor-subset} we have $\taylor A\subseteq\taylor{M'}$. We conclude by
	observing that terms in $\taylor A$ are in normal form, and that normal terms in
	$\taylor {M'}$ must also belong to $\nf{\taylor{M'}}$ as they are not affected
	by any reduction.
\end{proof}

\begin{proof}[Proof of Lemma~\ref{lem:taylor-bohm:nf-imply-in-approx}]
	By induction on $m$.
	\begin{itemize}
		\item If $m=x$, then $M=x$ and we set $A=x$.
		\item If $m=\lambda xn$, then $M=\lambda xN$ with $n\app N$.
			Since $n$ must be in normal form, we can apply the
			induction hypothesis to obtain some $A'\smaller N$
			such that $n\app A'$. We then set $A=\lambda xA'$.
		\item If $m=m_1m_2$, then $M=M_1M_2$ with $m_i\app M_i$. Again,
			by induction hypothesis, we have $A_1\smaller M_1$
			and $A_2\smaller M_2$ with $m_i\app A_i$. It remains
			to show that $A_1A_2$ belongs to the set of approximants
			described in Definition~\ref{def:bohm_approx}. Notice
			that $m_1$ cannot be of shape $l\<\lambda x—\>$, then since $m_1\app
			A_1$, $A_1$ cannot be a bottom or an abstraction. A
			simple examination of the syntax of approximants is
			enough to conclude that $A_1A_2$ indeed belongs to it.
		\item If $m=\der n$, then we again obtain $n\app N$ and some
			$A'\smaller N$ with $n\app A$. Since $n$ cannot be of
			shape $l\<[-]\>$.
			 (since $m$ is normal), we can again check
			Definition~\ref{def:bohm_approx} to conclude that $A=\der
			A'$ is an approximant, and that $m\app A$.
		\item If $m=n\es px$, we reason as for the application case, but
			using the case that $p$ cannot be of shape $l\<[-]\>$. 
		\item If $m=[n_1,\dots ,n_k]$, then $M=\oc N$ with $n_i\app N$ for
			all $i\in\{1,\dots ,k\}$. Then by induction hypothesis, there
			is $A_i\smaller N$ with $n_i\app A_i$. We then take
			$A=\oc A_i$, which works for any $i$.\qedhere
	\end{itemize}
\end{proof}

\begin{proof}[Proof of Theorem~\ref{thm:taylor-bohm-commute}]
We proceed by double inclusion.
\begin{itemize}
	\item Take $m \in \taylor{ \btb{M}}$, then there exists some $A_0\in\setapprox(M)$ such that
		$m \in \taylor{ A_0}$, by Definition~\ref{def:taylor_bohm}.
		There is $M_0$ such that $M \todbf^* M_0$ and $A_0 \smaller M_0$.
		We can therefore apply Lemma
		\ref{lem:approx_subset_nft} to conclude
		that $m \in \taylornf {M_0}$, which is
		equal to $\taylornf M$ by Lemma \ref{lem:bohm-taylor:equiv-imply-same-nfTaylor}.
	\item Assume $m \in \taylornf M$. By Lemma
		\ref{lem:taylor:approx-in-reduces-bis} there exists $M_0$ such
		that
		$M \todbf^* M_0$ and $m \in \taylor{M_0}$. By Lemma
		\ref{lem:taylor-bohm:nf-imply-in-approx}, there is $A\smaller
		M_0$ such that
		$m \in \taylor{ A}$. By definition, $A\in\setApprox(M)$, so we conclude that
		$m \in \taylor{ \btb{M}}$. \qedhere
\end{itemize}
\end{proof}

  \section{Proofs of Section~\ref{sec:translations}}
\label{appendix:translation}

\begin{proof}[Proof of Lemma~\ref{lem:taylor_trad_dcbn_dcbv}]
	For the $\dcbn$ case:
	Considering a resource term $m\in\dbangres$, we can show that $m\appn M$ (see
	Definition~\ref{def:dcbn_taylor_approx}) if and only if $m\app \tradn
	M$ (see Figure~\ref{fig:approx_dbang}), by an immediate induction on $M$.
	\begin{itemize}
		\item $x$ is the only approximation of $x=\tradn x$.
		\item $m\appn\lambda xN \in \dcbn$ if and only if $m=\lambda xn$ with
			$n\appn N$, if and only if (induction hypothesis) $n\app \tradn N$,
			and then if and only if $m\app \lambda x\tradn N=\tradn{(\lambda xN)}$
		\item $m\appn NP\in\dcbn$ if and only if $m=n[p_1,\dots ,p_k]$ with
			$n\appn N, p_i\app P$ for any $i\leq k$ if and only if
			(induction hypothesis) $n\app \tradn N$ and
			$p_i\app\tradn P$, and then if and only if $m\app
			\tradn{(NP)}$.
		\item Case $M= N[P/x]$ is similar to the previous one. \qedhere
	\end{itemize}

	For the $\dcbv$ case:
		By induction on $M$:
	\begin{itemize}
		\item For any $k\in\mathbb N$, $[x]_k\appv x$ and $[x]_k\app
			\tradv x = \oc x$.
		\item For any $k\in\mathbb N$, $[\lambda xm_1,\dots ,\lambda xm_k]
			\appv \lambda xM$ iff $m_i\appv M$ for all $i$ iff
			(induction hypothesis) $m_i\app \tradv M$ iff $[\lambda
			xm_1,\dots ,\lambda xm_k]\app\tradv{(\lambda xM)}=\oc(\lambda
			x\tradv M)$.
		\item $m\appv NP$.
			\begin{itemize}
				\item Either $N$ is an application. Then
					$m\appv NP$ iff $m=\der np$ with $n\appv N$ and $p\appv
					P$ iff (induction hypothesis) $n\app \tradv N$ and
					$p\app\tradv P$ iff $\der np\app\tradv{(NP)}$.
				\item Either $N=\oc N'$, and then $m\appv NP$
					iff $m= n'p$ with $n'\appv N'$ and
					$p\appv P$ iff $n'\app\tradv{N'}$, and
					$p\app\tradv P$ iff
					$n'p\app\tradv{(NP)}$. \qedhere
			\end{itemize}
	\end{itemize}
\end{proof}

\begin{lemma}[Substitution]\label{lem:trad_subst}~
	\begin{enumerate}
		\item Let $M,N\in\dcbn$. $\tradn M\{\tradn
			N/x\}=\tradn{M\{N/x\}}$.
		\item Let $M,N\in\dcbv$. If $\tradv N=\oc P$, then
			$\tradv{M\{N/x\}}=\tradv M\{P/x\}$
	\end{enumerate}
\end{lemma}
\begin{proof}
	Standard induction on $M$. This substitution lemma only holds for value in the
	case of $\dcbv$, because substituting a term to (the traduction of) a variable
	necessarily puts it under an exponential. This is not an issue because in \dcbv, these
	substitutions occur only if it is the case.
\end{proof}

\begin{proof}[Proof of Theorem~\ref{lem:trad_embedding}]
	For $\circ\in\{v,n\}$, the statements of the lemma can be depicted as
	follows, where the dashed lines and the term $P$ are the ones to be
	established:
	\begin{center}\begin{tikzpicture}[scale=0.7]
		\node(mn) at (0,0){$M^\circ$};
		\node(n) at (3,0){$N$};
		\node(m) at (0,-2){$M$};
		\node(p) at (5,-2){$P$};
		\node(pn) at (5,0){$P^\circ$};

		\draw[->](mn)--node[pos=1,below]{\tiny{$\oc$}}(n);
		\draw[->, dashed](m)--node[pos=1,above]{$*$}node[pos=1,below]{$\circ$}(p);
		\draw[->](m)--node[left]{$()^\circ$}(mn);
		\draw[->](p)--node[left]{$()^\circ$}(pn);
		\draw[->,dashed](n)--node[pos=1,above]{$*$}node[pos=1,below]{\tiny{$\oc$}}(pn);
	\end{tikzpicture}\end{center}
	First, notice that the translations $\tradv{()}$ and $\tradn{()}$
	generate only redexes of shape $L\<\lambda xM\>N$ and
	$M[L\<\oc N\>/x]$
	\footnote{
		Redexes like $\der{L\<\oc N\>}$ can however
		appear during reductions, from translations $\tradv{()}$, but not in the
		translation itself.
	}.

	(1) By induction on the reduction $\tradn M\todbf N$.
	\begin{itemize}
		\item $\tradn M=\tradn L\<\lambda x \tradn{N_1}\> \oc\tradn{N_2}$ and
			$N=\tradn{L}\<\tradn{N_1}[\oc\tradn{N_2}/x]\>$. We have
			$N=\tradn{(L\<N_1[N_2/x]\>)}$ by definition of
			$\tradn{()}$. 
			Since $M=L\<\lambda x N_1\>N_2$, we have
			$M\todbf L\<N_1[N_2/x]\>$, and we are done, setting
			$P=L\<N_1[N_2/x]\>$.

		\item $\tradn M=\tradn{N_1}[\tradn L\<\oc \tradn{N_2}\>/x]$ and
			$N=\tradn L\<\tradn{N_1}\{\tradn{N_2}/x\}\>$. By
			Lemma~\ref{lem:trad_subst},
			$N=\tradn{(L\<N_1\{N_2/x\}\>)}$, and again we are done,
			setting $P=L\<N_1\{N_2/x\}\>$, since $M=N_1[L\< N_2\>/x]\todbf L\<N_1\{N_2/x\}\>$.
		\item The reduction is contextual:
			\begin{itemize}
				\item $\tradn M=\tradn{N_1}\oc\tradn{N_2}$ and
					$N=N_1'\oc\tradn N_2$ with $N_1\ton
					N_1'$. By induction hypothesis, there is
					some $P_1$ such that $N_1\ton^* P_1$ and
					$N_1'\todbf^* \tradn{P_1}$. Then we set
					$P=P_1N_2$, and we have indeed
					$M=N_1N_1\ton^*P$ and $N=N_1'\oc\tradn
					N_2\to^*
					\tradn{P_1}\oc\tradn{N_2}=\tradn{(P_1N_2)}=\tradn P$.
				\item $\tradn M=\tradn{N_1}\oc\tradn{N_2}$ and
					$N=\tradn{N_1}\oc N_2'$ with $N_2\todbf
					N_2'$. By induction hypothesis, we have
					some $P_2$ such that $N_2\ton^*P_2$ and
					$N_2'\todbf^*\tradn{P_2}$. We then set
					$P=N_1P_2$.
				\item $\tradn
					M=\tradn{N_1}[\oc\tradn{N_2}/x],N=\tradn{N_1}[\oc
					N_2'/x]$ with $N_2\todbf N_2'$. We reason
					as in the previous case.
				\item $\tradn M=\lambda x\tradn N_0, N=\lambda x
					N_0'$ with $\tradn{N_0}\todbf N_0'$. By induction
					hypothesis, there is $P_0$ such that
					$N_0\ton^* P_0$ and $N_0'\todbf^*
					\tradn{P_0}$. We then set $P=\lambda x
					P_0$.
			\end{itemize}
	\end{itemize}
	(2) By induction on the reduction $\tradv M\todbf N$. The first two
	cases are similar to before except the position of the exponential.
	\begin{itemize}
		\item $\tradv M=\tradv L\<\lambda x \tradv{N_1}\>\tradv{N_2}$ and
			$N=\tradv L\<\tradv{N_1}[\tradv{N_2}/x]\>$. We have
			$N=\tradv{(L\<N_1[N_2/x]\>)}$ by definition of
			$\tradv{()}$, and $M=L\<\lambda xN_1\>N_2$. We set
			$P=L\<N_1[N_2/x]\>$, satisfying $M\tov P$ and $N\todbf^0
			\tradv P$.
		\item $\tradv M=\tradv{N_1}[\tradv L\<\oc \tradv{N_2}\>/x]$ and
			$N=\tradv L\<\tradv{N_1}\{\tradv{N_2}/x\}\>$, then
			$M=N_1[L\<N_2'\>/x]$ with $N_2'$ being a
			variable or an application. By
			Lemma~\ref{lem:trad_subst}
			$\tradv{N_1\{N_2'/x\}}=\tradv{N_1}\{\tradv{N_2}/x\}$, so
			$N=\tradv{(L\<N_1\{N_2/x\})}$. We then set
			$P=L\<N_1\{N_2/x\}\>$.
		\item The reduction is contextual. We only detail the case where
			the reduction occurs in the left member of an
			application and under a dereliction; the other
			cases follow from an application
			of the induction hypothesis as before. The second of
			these two cases is
			important, as it is the responsible for the only
			configuration where $P$ must be distinct from $N$.
			\begin{itemize}
				\item $\tradv M=\der{\tradv N_1}\tradv N_2$ (in
					that case $M=N_1N_2$ with $N_1$ not
					being of shape $L\<V\>$ for any value
					$V$) and $N=\der{N_1'}\tradv{N_2}$. By
					induction hypothesis, there is some
					$P_1$ such that $N_1\tov^*P_1$ and
					$N_1'\todbf^*\tradv{P_1}$. We then have two
					possibilities:
					\begin{itemize}
						\item $\tradv{P_1}\neq L\<!-\>$
							($P_1$ is not a value).
							Then,
							$\tradv{(P_1N_2)}=\der{\tradv{P_1}}\tradv{N_2}$.
							We then set $P=P_1N_2$,
							and we have
							$M=N_1N_2\tov^*P_1N_2$
							and
							$N\todbf^*\der{\tradv{P_1}}\tradv{N_2}=\tradv
							P$.
						\item $\tradv{P_1}=\tradv L\<\oc
							\tradv Q\>$. Then,
							$\tradv{(P_1N_2)}=
							\tradv L\<\tradv Q\>\tradv
							N_2$. We have $N\todbf^*\der{\tradv
                                                       L\<\oc\tradv
                                                       Q\>}\tradv{N_2}\todbf\tradv L\<\tradv Q\>\tradv{N_2}$ 
                                                       by a single reduction
                                                       step\footnote{
                                                               These steps are called \emph{administrative}
                                                               in Arrial, Guerrieri and Kessner's
                                                               work~\cite{arrial-diligence}.
                                                       }. We then set $P=L\< Q\>N_2$.
                                                       It verifies $N\todbf^*\tradv P$ as we
                                                       just saw. We also have
							$M\tov^* P$, because
                                                       $M=N_1N_2$ and $N_1\tov^*P_1$.
                                                       Since $\tradv{P_1}=\tradv
                                                       L\<\oc\tradv Q\>$, it follows that
                                                       $Q$ is a value (either a variable
                                                       or an abstraction) and that
                                                       $P_1=L\<Q\>$, by definition of
                                                       $\tradv{()}$. \qedhere
					\end{itemize}

			\end{itemize}
	\end{itemize}
\end{proof}

\begin{proof}[Proof of Theorem~\ref{thm:bohm-trad}]
	For this proof we will benefit from the properties of Böhm trees stated
	in Proposition~\ref{prop:bohm-properties}.

	(1) We proceed by coinduction on $\btn(M)$. 
	\begin{itemize}
		\item If $\btn(M)=\bot$, then $\setapprox (M)=\{\bot\}$. We need
			to show that $\setapprox(\tradn M)=\{\bot\}$. Consider
			$A\in\setapprox(\tradn M)$, we have $\tradn M\todbf^*N$
			with $A\sqsubseteq N$. By
			Lemma~\ref{lem:trad_embedding}, we have some $P$ such
			that $M\todbf^*P$ and $N\todbf^*\tradn P$. By
			Lemma~\ref{lem:bang:bohm:reduce-have-at-least-same-approxs},
			$A\smaller \tradn P$. Now, observe that $P$ must be some
			$\dcbn$ redex, otherwise $\setapprox(M)$ would contain
			other approximations than $\bot$. Then, by simulation
			(Theorem~\ref{thm:trad_meaningfulness}), $\tradn P$ also
			is a redex, and since the syntax of approximants
			contains no redex, necessarily $A=\bot$. We conclude
			that $\btb(\tradn M)=\bot$.
		\item If $\btn(M)=x$, then
			$\tradn{\btn(M)}=x=\btb(x)=\btb(\tradn x)$.
		\item If $\btn(M)=\lambda x\btn(N)$, then
			$\tradn{(\btn(M))}=\tradn{(\lambda x\btn(N))}=\lambda
			x\tradn{(\btn(N))}$ (by definition of $\tradn{()}$). By
			coinduction hypothesis, $\tradn{(\btn(N))}=\btb(\tradn
			N)$. Then, $\tradn{(\btn(M))}=\lambda x\btb(\tradn
			N)=\btb(\lambda x\tradn N)=\btb(\tradn{(\lambda x
			N)})=\btb(\tradn M)$.
		\item If $\btn(M)$ is an application, then it is equal to
			some $(x)\btn(N)$. Then, we have
			$\tradn{(\btn(M))}=(x)\oc\tradn{(\btn(N))}=
			(x)\oc\btb(\tradn N)$, by coinduction hypothesis,
			which is equal to $\btb((x)\oc \tradn
			N)=\btb(\tradn{(x)N})=\btb(\tradn M)$.
		\item $\btn(M)$ cannot contain any explicit substitution, as they
			always correspond to redexes in $\dcbn$.
	\end{itemize}
	
	(2) We proceed by coinduction on $\btv(M)$.
	\begin{itemize}
		\item If $\btv(M)=\bot$, we reason as above, using this time the
			second item of Lemma~\ref{lem:trad_embedding}.
		\item If $\btv(M)=x$, then $\tradv{(\btv(M))}=\tradv x=\oc
			x=\btb(\oc x)=\btb(\tradv x)$.
		\item If $\btv(M)=\lambda x\btv(N)$, then
			$\tradv{(\btv(M))}=\oc(\lambda x\tradv{(\btv(N))})$. By
			coinduction hypothesis, it is equal to $\oc(\lambda
			x\btb(\tradv N))=\btb(\oc(\lambda x\tradv
			N))=\btb(\tradv{(\lambda xN)})=\btb(\tradv M)$.
		\item If $\btv(M)$ is an application, then it must be equal to
			some $x([\btv(N_i)/y_i])_{1\leq i\leq k}\btv(\tradv N_0)$
			(because in this case $M$ reduces to some $L\<x\>N_0$).
			
			Then $\tradv{(\btv(M))}=
			x([\tradv{(\btv(N_i))}/y_i])_{1\leq i\leq k}\tradv{(\btv(N_0))}$.
			By coinduction hypothesis,
			$\tradv{(\btv(N_j))}=\btb(\tradv{N_j})$ for
			$j\in\{0,\dots ,k\}$. Then
			$\tradv{(\btv(M))}=
			x([(\btb(\tradv{N_i}))/y_i])_{1\leq i\leq
			k}\btb(\tradv{N_0})$. Again, for $i\in\{1,\dots ,k\}$,
			$\btb(\tradv{N_i})$ cannot be an exponential, since
			those explicit substitutions must not be reducible. Then,
			$\tradv{(\btv(M))}=\btb\left(x\left([\tradv{N_i}/x]\right)_{i\in\{1,\dots ,k\}}\tradv
			N_0\right)=\btb(\tradv M)$.
		\item If $\btv(M)=\btv(N_1)[\btv(N_2)/x]$, then again $\btv(N_2)$
			cannot be a value, hence
			$\tradv{(\btv(M))}=\btb(\tradv{N_1})[\btb(\tradv{N_2})/x]$
			(by coinduction hypothesis)
			$=\btb(\tradv{N_1}[\tradv{N_2}/x])=\btb(\tradv{(N_1[N_2/x])})$. \qedhere
	\end{itemize}
\end{proof}

\begin{proof}[Proof of Theorem~\ref{commutation-thm-translations}]
	Let $\circ\in\{n,v\}$. We have the following equalities:
	\begin{align*}
		\nf{\mathcal T^\circ (M)} &= \nf{\taylor{M^\circ}} &\text{By
		Lemma~\ref{lem:taylor_trad_dcbn_dcbv}} \\
		&= \taylor{BT(M^\circ)}   &\text{By
		Theorem~\ref{thm:taylor-bohm-commute}}  \\
		&= \taylor{(BT_\circ(M)}^\circ) &\text{By
		Theorem~\ref{thm:bohm-trad} } \\
		&=\mathcal T_\circ(BT_\circ(M)) &\text{By
		Lemma~\ref{lem:taylor_trad_dcbn_dcbv}}
	\end{align*}
\end{proof}

\section{Proofs of Section~\ref{sec:meaningfulness}}
\label{appendix:meaningfulness}

\begin{proof}[Proof of Theorem~\ref{thm:tnf-dbang}]
	We show the contrapositive of the statement: assume that $\taylornf
	M=\emptyset$. By definition, for every $m\in\taylor M$, we have $m\todbr^*
	\emptyset$. Consider now any resource testing context $t$. We can easily
	establish that $t\<m\>\todbr^* \emptyset $; since $\emptyset m=\emptyset,
	(\lambda x \emptyset)m=\emptyset$, and then by induction.

	For $M$ to be meaningful, there must exist a testing context $T$ such that
	$T\<M\>\todbs^k \oc P$ for some $P$ and $k\in\mathbb N$. We have
	$[]\app \oc P$. By iteratively applying 
	Lemma~\ref{lem:simu_dbr_bis}, there is some term $s\app T\<M\>$ such
	that $s\todbrs^k []$. By Lemma~\ref{lem:approx_contexts}, and because
	testing contexts are included in surface contexts, we have some $t\app
	T, m\app M$ such that $s=t\<m\>$.

	This leads to a contradiction: $t\<m\>\todbrs []$, yet we have shown
	that $t\<m\>\todbrs \emptyset$ for any $t$.
\end{proof}

\begin{proof}[Proof of Lemma~\ref{lemma:dbangn-nf-shape}]
	The following observations suffice:
	\begin{itemize}
		\item terms in \dbangn{} containing explicit substitutions are
			reducible. 
		\item If the leftmost subterm of an application is not a
		variable, the entire term is reducible. \qedhere
	\end{itemize}
\end{proof}

\begin{proof}[Proof of Theorem~\ref{thm:dbangn_tnf}]
	Consider some
	$p\in\taylornf M$, assumed non-empty. Then there is some $m\app M$ such
	that $m\todbrf^*p$.

	By Lemma~\ref{lem:dbangn_nf}, $p=\lambda x_1\dots x_k(x)\bar n_1\dots \bar n_l$
	for some $k,l\in\mathbb N$.
	Proposition~\ref{prop:standardisation_resources} allows us to focus on
	surface reduction: there are some $m'=\lambda
	x_1\dots x_k(x)\bar
	n_1',\dots ,\bar n_l'$ such that $m\todbrs^* m'\totodbr^* p$ (the second part
	of the reduction acting inside the bags).

	Then, by iteratively applying Lemma~\ref{lem:simu_dbr_n}, we obtain $M'\in\dbang$
	such that $m'\app M'$ and $M\todbs^*M'$. By the definition of
	the approximation relation $\app$, we have $M'=\lambda x_1\dots x_k(x)\oc
	N'_1\dots \oc N'_l$ where $\bar n_i'\app \oc N'_i$.

	We now define the appropriate testing context $T= ((\lambda x
	\square) \oc(\lambda y_1\dots y_l\oc z_0))\oc z_1\dots \oc z_k$
	where the $z_i$ are chosen distinct from the $x_j$ and $y_j$.

	We observe that $T\<M'\>\todbs (\lambda x_1\dots x_k(\lambda y_1\dots y_l\oc
	z_0)\oc z_1\dots \oc z_k)\oc N_1'\dots \oc N_k'\todbs^2\oc z_0$.
	We conclude as follows: since $T\<M'\>\todbs^2\oc z_0$, $M\todbs^*M'$ and $T$ is a surface
	context, we have $T\<M\>\todbs^*\oc z_0$ by the contextual closure of
	$\todbs$. Therefore, $M$ is meaningful.
\end{proof}

\begin{proof}[Proof of Corollary~\ref{cor:dcbn_meaningful}]
	Recall that, by Corollary~\ref{cor:tradn_tnf-tradv-tnf} and Theorem~\ref{thm:trad_meaningfulness}, 
	$\taylornf M\neq\emptyset$ if and
	only if $\taylornf {\tradn M}\neq\emptyset$, and $M$ is meaningful if
	and only if $\tradn M$ is meaningful.

	($\to$) Assume $M$ is meaningful, then $\tradn M$ is also meaningful, and
	by Theorem~\ref{thm:tnf-dbang}, $\taylornf{\tradn M}\neq \emptyset$.
	It follows that $\taylornf M\neq\emptyset$.

	($\leftarrow$) If $\taylornf M\neq\emptyset$, then $\taylornf{\tradn
	M}\neq\emptyset$. By Lemma~\ref{lem:trad_fragments}, $\tradn
	M\in\dbangn$, and Theorem~\ref{thm:dbangn_tnf} implies that $\tradn M$
	must be meaningful. Therefore, $M$ is also meaningful.
\end{proof}

In what follows, the variables in $\circ_i$ are always taken fresh, so as they
do not interfere with variables in the terms where the $\circ_i$ are
substituted. In particular, we use the fact that for any $i>1$ and any $M$, $\circ_k \oc
M\todbf^2\oc\circ_{k-1}$.

\begin{proof}[Proof of Lemma~\ref{lem:mutual1}]
	By induction on the syntax $B$:
	\begin{itemize}
		\item If $M$ is of the form $\oc x$ or $\oc\lambda xB$, then we are
			done since $M\sigma\todbf^0 \oc P$ for some $P$ and any
			substitution $\sigma$.
		\item If $M\in B_\oc$, we apply Lemma~\ref{lem:mutual2},
			which guarantees the existence of a substitution $\sigma$
			such that $M\sigma\todbf^*\oc\circ_j$ for some $j$.
		\item If $M=N[P_1/x_1]\dots [P_k/x_k]$, then $N\in B$ and $P_i\in
			B_\oc$ for all $i$. Let $\{y_1,\dots ,y_l\}=\fv(N)\cup\bigcup_{i\leq
			k}\fv(P_k)$. By induction hypothesis, we have
			$c$ such that for any $k_1,\dots ,k_l\geq c$,
			$N\{\circ_{k_1}/y_1,\dots ,\circ{k_l}/y_l\}\todbf \oc P$ for
			some $P$. By Lemma~\ref{lem:mutual2},
			for $i\leq k$, there exist $c_i$
			and $n_i$ such that for any $k_{i,1},\dots ,k_{i,l}\geq c_i$,
			and for some $j_i\geq n_i$,
			$P_i\{\circ_{k_{i,1}}/y_1,\dots ,\circ_{k_{i,l}}/y_l\}\todbf^*\oc\circ_{j_i}$.

			Consider then $m_i=\max\{k_i,k_{i,1},\dots k_{i,k}\}$ for any
			$i\leq l$. We then have some $P'$ and some $r_i$ with
			$N[P_1/x_1]\dots [P_k/x_k]\{\circ_{m_1}/y_1,\dots ,\circ_{m_l}/y_l\}\todbf^*
			\oc P' [\oc\circ_{r_1}/y_1]\dots [\oc\circ_{r_l}/y_l]\todbf^l 
			\oc P' \{\oc\circ_{r_1}/y_1\}\dots \{\oc\circ_{r_l}/y_l\}$,
			which is again a value as required. \qedhere
	\end{itemize}
\end{proof}

\begin{lemma}
\label{lem:mutual2}
	Let $\{x_1,\dots ,x_n\}$ be a set of variables and $M\in B_\oc$
	(Definition~\ref{def:syntax_nf_dbangv}) with
	$\fv(M)\subseteq\{x_1,\dots ,x_n\}$. There exist $k,c\in\mathbb N$ such
	that for any $k_1,\dots ,k_n\geq c$, there is some $j\geq k$ with
	$M\{\circ_{k_1}/x_1\dots \circ_{k_n}/x_n\}\todbf^* \oc \circ_j$. 
\end{lemma}
\begin{proof}
	By induction on $B_\oc$:
	\begin{itemize}
		\item $M=L\<x\>N=(x[N_1/y_1]\dots [N_m/y_m])N$ with $N\in B$, $N_i\in
			B_\oc$ for all $x\leq m$, and $\{x_1,\dots ,x_n\}=\fv(M)$. By
			Lemma~\ref{lem:mutual1}, there
			exists $k$ such that
			$N\{\circ_{k_1}/x_1,\dots ,\circ_{k_n}/x_n\}\todbf^* \oc P$ for
			any $k_i\geq k$ and some $P$. 
			
			By induction hypothesis, we have for each $i\leq m$,
			some $l_i$ and $c_i$ such that for any $l_{i,1},\dots ,l_{i,n}\geq l_i$
			$N_i \{\circ_{l_{i,1}}/y_1,\dots ,\circ_{l_{i,n}}/y_n\}\todbf^*
			\oc\circ_{j_i}$ for some $j_i\geq c_i$.

			Let $n_x$ be the index of $x$ in
			$\{x_1,\dots ,x_n\}$ (of course, $x\in\fv(M)$).

			We then
			set $r_i=\max\{k_i,
			l_{i,1},\dots ,l_{i,n}\}$ for each $i\neq n_x$; and we
			consider $r_{n_x}$ an arbitrary integer greater or equal
			to $\max\{k_{n_x},l_{n_x,1},\dots ,l_{n_x,n}\}$.
			
			We find 
			that $M\{\circ_{r_1}/x_1,\dots ,\circ_{r_m}/x_m\}\todbf^*
			\circ_{n_x}[\oc\circ_{r_1'}/x_1]\dots [\oc\circ_{r_m'}/x_1]\oc
			P'$, with $r'_i\geq
			r_i$ for all $i\leq m$. The reduction then yields 
			$\circ_{n_x}\oc P'$, which reduces immediately to $\oc
			\circ_{r_{n_x}-1}$. This concludes the case, as the
			reduction holds for any
			$r_{n_x}\geq\max\{k_{n_x},l_{n_x,1},\dots ,l_{n_x,n}\}$.
			
		\item $M=\der N N'$ with $N\in B_\oc$, $N'\in B$, and
			$\{x_1,\dots ,x_m\}=\fv(M)$. By induction
			hypothesis, we have some $n, n',c$ such that for any
			$n_1,\dots ,n_m\geq n$ and ,
			$N\{\circ_{n_1}/x_1,\dots ,\circ_{n_m}/x_m\}\todbf^*\oc\circ_j$ for
			all $j\geq c$, and for any $n'_1,\dots ,n'_m\geq n'$,
			$N'\{\circ_{n'_1}/x_1,\dots ,\circ_{n'_m}/x_m\}\todbf^*\oc P$
			for some $P$.

			Then, consider $k_i=\max\{n_i,n_i'\}$ for any $i\leq m$,
			we have that
			$M\{\circ_{k_1}/x_1,\dots ,\circ_{k_m}/x_m\}\todbf^* \der{\oc
			\circ_j}\oc P'\todbf \circ_j\oc P'\todbf^2\oc\circ_{j-1}$
			(notice that we need here to take $j>1$, which is
			allowed by our hypothesis).
		\item $M=N[P_1/x_1]\dots [P_k/x_k]$. This case is similar to the third case of
			Lemma~\ref{lem:mutual1}: the explicit substitutions are
			removed after an application of the induction
			hypothesis. \qedhere
	\end{itemize}
\end{proof}

\begin{proof}[Proof of Theorem~\ref{thm:dbangv_tnf}]
	Consider some
	$p\in\taylornf M$, assumed non-empty. There is some $m\app M$ such
	that $m\todbrf^*p$. By Lemma~\ref{lem:syntax_nf_dbangvres}, $p$ belongs to the syntax $b$.
	Proposition~\ref{prop:standardisation_resources} ensures that there is some $p'\in
	b$ such that $m\todbrs p'$ (as in Theorem~\ref{thm:dbangn_tnf}, we
	focus on internal reduction). 

	By iteratively applying Lemma~\ref{lem:simu_dbr_bis}, we obtain $P'$ such that
	$M\todbf^* P'$ and $p'\app P'$. By the definition of $\app$, we also
	have that $P'\in B$.

	Let $\{x_1,\dots ,x_k\}=\fv(P')$. Lemma~\ref{lem:mutual1} implies that there
	are some terms $N_1,\dots ,N_k$ such that $P'\{N_1/x_1,\dots ,N_k/x_k\}\todbf^* \oc
	Q$ for some term $Q$. 

	We define the testing context $C=(\lambda x_1\dots \lambda x_k \square)\oc
	N_1\dots \oc N_k$, which satisfies $C\<P'\>\todbf^*\oc Q$.
	We can conclude that $M$ is meaningful by the contextuality of reduction,
	since $C\<M\>\todbf^* C\<P'\>\todbf^* \oc Q$.
\end{proof}

\begin{proof}[Proof of Corollary~\ref{cor:dcbv-meaningful-iff-taylor}]
	Recall that, by Corollary~\ref{cor:tradn_tnf-tradv-tnf} and Theorem~\ref{thm:trad_meaningfulness}, 
	$\taylornf M\neq\emptyset$ if and
	only if $\taylornf {\tradv M}\neq\emptyset$, and $M$ is meaningful if
	and only if $\tradv M$ is meaningful.

	($\to$) Assume $M$ is meaningful, then $\tradv M$ is meaningful, and
	by Theorem~\ref{thm:tnf-dbang}, $\taylornf{\tradv M}\neq \emptyset$.
	It follows that $\taylornf M\neq\emptyset$.

	($\leftarrow$) If $\taylornf M\neq\emptyset$, then $\taylornf{\tradv
	M}\neq\emptyset$. By Lemma~\ref{lem:trad_fragments}, $\tradv
	M\in\dbangv$, and Theorem~\ref{thm:dbangv_tnf} implies that $\tradv M$
	is meaningful. Therefore, $M$ is also meaningful.
\end{proof}

\fi
\end{document}